\documentclass[11pt,letterpaper]{article}

\usepackage[margin=1in]{geometry}
\usepackage{nicefrac}
\usepackage{amsmath}
\usepackage{amsthm}
\usepackage{amssymb}
\usepackage{mathtools}
\usepackage{dsfont}
\usepackage[dvipsnames]{xcolor}
\usepackage{enumitem}
\usepackage{amsfonts}
\usepackage{enumitem}
\usepackage{mathtools}
\usepackage{comment}
\usepackage{url}
\usepackage{titlesec}

\allowdisplaybreaks
\titlelabel{\thetitle.\quad}
\titlespacing{\paragraph}{%
	0pt}{
	0\baselineskip}{
	1em}

\newcommand{\Omit}[1]{}

\newcommand{\matt}[1]{{\color{blue}{MW: #1}}}

\newcommand{\mattnote}[1]{{\color{blue}{#1}}}

\newtheorem{theorem}{Theorem}[section]
\newtheorem*{theorem_nonum}{Theorem}
\newtheorem*{theorem_main}{Main Theorem}
\newtheorem*{result_main}{Main Result}
\newtheorem{lemma}[theorem]{Lemma}

\newtheorem{proposition}[theorem]{Proposition}
\newtheorem{claim}[theorem]{Claim}

\newtheorem{corollary}[theorem]{Corollary}
\newtheorem{observation}[theorem]{Observation}

\newtheorem{fact}[theorem]{Fact}
\newtheorem{definition}{Definition}[section]

\theoremstyle{definition}
\newtheorem{defin}{Definition}[section]

\newcommand{\poly}{\text{poly}}

\newcommand{\abs}[1]{\vert{#1}\vert}
\newcommand{\efs}{\textsc{Exist-Far-Sets}}
\newcommand{\fs}{\textsc{Far-Sets}}
\newcommand{\equ}{\textsc{Equality}}
\newcommand{\disj}{\textsc{Disjointness}}
\newcommand{\ghd}{\textsc{Gap-Hamming-Distance}}
\newcommand{\wm}{\textsc{Welfare-Maximization}}

\newcommand{\EE}{\mathbb{E}}

\newcommand\norm[1]{\left\lVert#1\right\rVert}

\newcommand\floor[1]{\left\lfloor#1\right\rfloor}
\newcommand\parens[1]{\left(#1\right)}
\newcommand*{\pr}[2][]{\text{Pr}\ifx\\\left[#1\right]\\\else_{#1}\fi \left[#2\right]}

\begin{document}


\newcommand{\AutoAdjust}[3]{\mathchoice{ \left #1 #2  \right #3}{#1 #2 #3}{#1 #2 #3}{#1 #2 #3} }
\newcommand{\Xcomment}[1]{{}}
\newcommand{\eval}[1]{\left.#1\vphantom{\big|}\right|}
\newcommand{\inteval}[1]{\Big[#1\Big]}
\newcommand{\InParentheses}[1]{\AutoAdjust{(}{#1}{)}}
\newcommand{\InBrackets}[1]{\AutoAdjust{[}{#1}{]}}
\newcommand{\Ex}[2][]{\operatorname{\mathbf E}_{#1}\InBrackets{#2\vphantom{E_{F}}}}
\newcommand{\Exlong}[2][]{\operatornamewithlimits{\mathbf E}\limits_{#1}\InBrackets{#2\vphantom{\operatornamewithlimits{\mathbf E}\limits_{#1}}}}
\newcommand{\Prx}[2][]{\operatorname{\mathbf{Pr}}_{#1}\InBrackets{#2}}
\newcommand{\Prlong}[2][]{\operatornamewithlimits{\mathbf Pr}\limits_{#1}\InBrackets{#2\vphantom{\operatornamewithlimits{\mathbf Pr}\limits_{#1}}}}
\def\prob{\Prx}
\def\expect{\Ex}
\newcommand{\super}[1]{^{(#1)}}
\newcommand{\dd}{\mathrm{d}}  
\newcommand{\given}{\;\mid\;}

\newcommand{\Ell}{\ensuremath{\mathcal{L}}}

\newcommand{\be}{\begin{equation}}
\newcommand{\ee}{\end{equation}}
\newcommand{\argmin}{\mathop{\rm argmin}}
\newcommand{\argmax}{\mathop{\rm argmax}}
\newcommand{\vr}[1]{{\mathbf{#1}}}
\newcommand{\bydef}{\stackrel{\bigtriangleup}{=}}
\newcommand{\eps}{\varepsilon}
\newcommand{\mm}[1]{\mathrm{#1}}
\newcommand{\mc}[1]{\mathcal{#1}}
\newcommand{\mb}[1]{\mathbf{#1}}
\newcommand{\vect}[1]{\ensuremath{\mathbf{#1}}}
\newcommand{\R}{\mathbb{R}}

\def \EE   {{\mathbb E}}
\def \OPT {\mathcal{OPT}}
\def \vf  {\textrm{vf}}
\def \dvf {\varphi}
\def \utility {u}
\def \reals {{\mathbb R}}

\newcommand{\Idr}[1]{\mathds{1}\InBrackets{#1\vphantom{\sum}}}


\newcommand{\dist}{\mathcal{F}}
\newcommand{\disti}[1][i]{{\mathcal{F}_{#1}}}
\newcommand{\distsmi}[1][i]{\dists_{\text{-}#1}}
\newcommand{\dists}{\vect{\dist}}
\newcommand{\distw}{\widetilde{\mathcal{F}}}

\newcommand{\valdist}{\mathcal{F}}
\newcommand{\valdists}{\vect{\valdist}}
\newcommand{\valdisti}[1][i]{{\valdist_{#1}}}
\newcommand{\valdistsmi}[1][i]{\valdists_{\text{-}#1}}

\newcommand{\dens}{f}
\newcommand{\denss}{\vect{ \dens}}
\newcommand{\densi}[1][i]{{\dens_{#1}}}

\newcommand{\agents}{N}
\newcommand{\nagent}{n}
\newcommand{\goods}{M}
\newcommand{\nitem}{m}
\newcommand{\auction}{A}

\newcommand{\weight}{w}
\newcommand{\partition}{\Gamma}
\newcommand{\partitioni}{\Gamma_i}

\newcommand{\CWE}[0]{\textsf{CWE}}
\newcommand{\ef}[0]{envy free}
\newcommand{\EF}[0]{EF}
\newcommand{\CWEWOMC}[0]{CWE without market clearance}
\newcommand{\SW}[0]{\textsf{SW}}
\newcommand{\Revenue}[0]{\textsf{Rev}}

\newcommand{\bid}{b}
\newcommand{\bids}{\vect{\bid}}
\newcommand{\bidsmi}[1][i]{\bids_{\text{-}#1}}
\newcommand{\bidi}[1][i]{{\bid_{#1}}}

\newcommand{\val}{v}
\newcommand{\vals}{\vect{\val}}
\newcommand{\valsmi}[1][i]{\vals_{\text{-}#1}}
\newcommand{\vali}[1][i]{{\val_{#1}}}
\newcommand{\valith}[1][i]{{\val_{(#1)}}}

\newcommand{\wal}{\widetilde{v}}
\newcommand{\wals}{\vect{\wal}}
\newcommand{\walsmi}[1][i]{\wals_{\text{-}#1}}
\newcommand{\wali}[1][i]{{\wal_{#1}}}

\newcommand{\util}{u}
\newcommand{\utils}{\vect{\util}}
\newcommand{\utili}[1][i]{\util_{#1}}
\newcommand{\utilsmi}[1]{\utils_{\text{-}#1}}

\newcommand{\price}{p}
\newcommand{\prices}{\vect{\price}}
\newcommand{\pricei}[1][i]{{\price_{#1}}}

\newcommand{\payment}{\rho}
\newcommand{\payments}{\vect{\payment}}
\newcommand{\paymenti}[1][i]{{\payment_{#1}}}

\newcommand{\type}{t}
\newcommand{\types}{\vect{\type}}
\newcommand{\typei}[1][i]{{\type_{#1}}}
\newcommand{\typesmi}[1][i]{\types_{\text{-}#1}}

\newcommand{\alloc}{X}
\newcommand{\allocs}{\vect{\alloc}}
\newcommand{\allocsmi}[1][i]{\allocs_{\text{-}#1}}
\newcommand{\alloci}[1][i]{{\alloc_{#1}}}

\newcommand{\opt}{\text{OPT}}
\newcommand{\opts}{\vect{ \opt}}
\newcommand{\opti}[1][i]{{\opt_{#1}}}
\newcommand{\optsmi}[1][i]{\opts_{\text{-}#1}}

\newcommand{\talloc}{Y}
\newcommand{\tallocs}{\vect{\talloc}}
\newcommand{\talloci}[1][i]{{\talloc_{#1}}}
\newcommand{\tallocsmi}[1][i]{\tallocs_{\text{-}#1}}

\newcommand{\rank}{r}
\newcommand{\ranks}{\vect{\rank}}
\newcommand{\ranki}[1][i]{{\rank_{#1}}}
\newcommand{\ranksmi}[1][i]{\ranks_{\text{-}#1}}

\newcommand{\crit}{\theta}
\newcommand{\crits}{\vect{\crit}}
\newcommand{\criti}[1][i]{{\crit_{#1}}}
\newcommand{\critsmi}[1][i]{\crits_{\text{-}#1}}

\newcommand{\decl}{d}
\newcommand{\decls}{\vect{\decl}}
\newcommand{\decli}[1][i]{{\decl_{#1}}}
\newcommand{\declsmi}[1][i]{\decls_{\text{-}#1}}

\newcommand{\demand}{D}
\newcommand{\demands}{\vect{\demand}}
\newcommand{\demandi}[1][i]{{\demand_{#1}}}
\newcommand{\demandsmi}[1][i]{\demands_{\text{-}#1}}

\newcommand{\sold}{{\texttt{SOLD}}}
\newcommand{\soldi}[1][i]{\sold_{#1}}

\newcommand{\CWEalg}{{\sc CWE} algorithm}

\newcommand{\Pool}{\text{Pool}}
\newcommand{\Pop}{\textbf{Pop}}
\newcommand{\Push}{\textbf{Push}}
\newcommand{\Bundle}{\textbf{Bundle}}
\newcommand{\ResolveConflict}{\textbf{ResolveConflict}}


\newcommand{\Reject}{\text{Reject}}
\newcommand{\AllocateDemand}{\textbf{AllocateDemand}}
\newcommand{\RaisePrices}{\textbf{RaisePrices}}

\title{Complement-Free Couples Must Communicate:\\A Hardness Result for Two-Player Combinatorial Auctions}

\author{
Tomer Ezra\thanks{Tel-Aviv University; {\tt tomer.ezra@gmail.com}}
\and
Michal Feldman\thanks{Tel-Aviv University; {\tt mfeldman@tau.ac.il}}
\and
Eric Neyman\thanks{Princeton University; {\tt eneyman@princeton.edu}}
\and
Inbal Talgam-Cohen\thanks{Technion; {\tt inbaltalgam@gmail.com}}
\and
Matt Weinberg \thanks{Princeton University; {\tt smweinberg@princeton.edu}}
}

\date{}

\maketitle

\begin{abstract}

\setlength{\parskip}{10pt}%
\setlength{\parindent}{0pt}%

We study the communication complexity of welfare maximization in combinatorial auctions with $m$ items and two subadditive bidders. A $\nicefrac{1}{2}$-approximation can be guaranteed by a trivial randomized protocol with zero communication, or a trivial deterministic protocol with $O(1)$ communication. \textit{We show that outperforming these trivial protocols requires exponential communication}, settling an open question of~\cite{Feige09}.

Specifically, we show that any (randomized) protocol guaranteeing a $(\nicefrac{1}{2}+\nicefrac{6}{\log_2 m})$-approximation requires communication exponential in $m$. This is tight even up to lower-order terms: we further present a $(\nicefrac{1}{2}+\nicefrac{1}{O(\log m)})$-approximation in $\poly(m)$ communication.


To derive our results, we introduce a new class of subadditive functions that are ``far from" fractionally subadditive functions, and may be of independent interest for future works.
Beyond our main result, we consider the spectrum of valuations between fractionally-subadditive and subadditive via the MPH hierarchy. Finally, we discuss the implications of our results towards combinatorial auctions with strategic bidders.

\end{abstract}

\addtocounter{page}{-1}
\newpage

\section{Introduction}
\label{sec:introduction}
{\bf Background.} Combinatorial auctions have been a driving force of Algorithmic Game Theory since its inception: how should one allocate goods among interested parties? That is, if a central designer has a set $M$ of $m$ indivisible goods to allocate, and each of $n$ players has a valuation function $v_i:2^M \rightarrow \mathbb{R}_+$ (private, known only to player $i$), we wish to partition the items to maximize the \emph{social welfare} ($\sum_i v_i(S_i)$, where $S_i$ denotes the items received by player $i$ in the partition).

This fundamental problem has received significant attention in various models: with or without incentives, with or without restrictions on valuations, with or without computational limits on the players, etc. In this paper we prove standard communication lower bounds when players have subadditive (also called complement-free) valuations.\footnote{A valuation function is subadditive if for all $S, T$ $v(S \cup T) \leq v(S) + v(T)$.} That is, our lower bounds rule out the existence of good mechanisms even when players honestly follow the intended protocol, are computationally unbounded, and are assumed to have subadditive valuation functions. The study of combinatorial auctions specifically through the lens of communication complexity has a rich history dating back to early works of Blumrosen, Nisan, and Segal~\cite{Nisan00, BlumrosenN02, NisanS06}, as such lower bounds sidestep challenging debates on appropriate behavioral assumptions. See Section~\ref{sec:related} for a high-level overview of this literature, and specifically the role of communication complexity.

\paragraph{State of the art.} On this front, the state-of-the-art is fairly remarkable: without any restrictions on the valuations, a $\max\{\nicefrac{1}{n},\nicefrac{1}{O(\sqrt{m})}\}$-approximation can be achieved in $\poly(n, m)$ communication~\cite{LaviS05}, and this is tight~\cite{NisanS06}. For fractionally subadditive valuations (also called XOS),\footnote{A valuation is fractionally subadditive if it can be written as a maximum over additive functions.} a $(1-(1-\nicefrac{1}{n})^n)$-approximation can be achieved in $\poly(n, m)$ communication~\cite{Feige09}, and this is tight~\cite{DobzinskiNS10}. For subadditive valuations, a $\nicefrac{1}{2}$-approximation can be achieved in $\poly(n, m)$ communication~\cite{Feige09}, and no better than a $(\nicefrac{1}{2}+\nicefrac{1}{2n})$-approximation can be achieved in $\poly(n, m)$ communication~\cite{DobzinskiNS10} (so this is tight as $n \rightarrow \infty$). As such, remaining open problems in this direction are scarce. The resolution of one such problem is the focus of this paper.

\paragraph{The case of $n=2$.}

While Feige's $\nicefrac{1}{2}$-approximation is tight as $n \rightarrow \infty$, the $n=2$ case was posed as an open problem in~\cite{DobzinskiNS10, Feige09}. It initially may seem unusual for the $n=2$ case to be singled out when the asymptotics are resolved, but there is substantial difference in the merits of a $\nicefrac{1}{2}$-approximation for $n=2$ and $n> 2$. Specifically, Feige's $\nicefrac{1}{2}$-approximation for $n > 2$ employs an incredibly sophisticated LP rounding, but the same guarantee is achieved by numerous trivial algorithms when $n=2$: (a) allocate all of $M$ to a uniformly random player, (b) allocate each item independently to a uniformly random player, (c) ask each player to report $v_i(M)$, and award $M$ to the highest bidder. Note that (a) and (b) are particularly trivial in that they are completely oblivious to the valuations. (a) and (c) are particularly trivial in that they maintain their guarantee even without subadditivity. All three can trivially be made into truthful auctions ((a) and (b) don't even solicit input, and are therefore truthful. (c) is simply a second-price auction on the grand bundle $M$). So resolving the gap between $\nicefrac{1}{2}$ and $\nicefrac{1}{2}+\nicefrac{1}{2n} = \nicefrac{3}{4}$ for $n=2$ is not just a question of determining the optimal constant, but really a question of whether it is possible to achieve any non-trivial guarantee. The main result of this paper answers no:

\begin{result_main}[Informal]
	For two subadditive players, the aforementioned trivial protocols ensuring a $\nicefrac{1}{2}$-approximation are optimal among those with subexponential communication.
\end{result_main}

\paragraph{Implications.} Before overviewing our construction and extensions, we wish to highlight two immediate implications of our results for combinatorial auctions with strategic bidders, via two recent reductions which renewed further interest specifically in the $n=2$ case.

(1) The power of truthful vs.~non-truthful communication-efficient protocols: The central driving theme of algorithmic mechanism design is understanding the relative power of truthful vs.~non-truthful ``efficient'' protocols~\cite{NisanR01}. Remarkably, when ``efficient'' refers to ``communication-efficient,'' no separation is known to exist --- for any valuation class or agent number --- despite significant gaps in the state-of-the-art approximation ratios (see Section~\ref{sec:related} for further detail, along with a brief discussion of related results concerning computational efficiency). Recent work of~\cite{Dobzinski16b} provides a deep structural connection between truthful communication-efficient combinatorial auctions and \emph{simultaneous} (non-truthful) communication-efficient combinatorial auctions \emph{specifically when $n=2$},%
\footnote{There are also implications when $n>2$, but not quite as strong as for $n=2$.}
thereby proposing extensive study of the $n=2$ case to search for the first separation. On this front, our result proves that in fact no separation exists for $n=2$ subadditive buyers, as the aforementioned trivial protocols (now proved to be optimal) are also truthful.

(2) Price of Anarchy of simple mechanisms: One measure by which the performance of (non-truthful) combinatorial auctions is quantified in strategic settings is the (Bayesian) price of anarchy (BPoA), defined as the worst ratio between the (expected) welfare of the worst equilibrium and the optimal (expected) welfare. For subadditive valuations, simultaneous \emph{first price} auctions are known to have BPoA at least $\nicefrac{1}{2}$ \cite{FeldmanFGL13}, and this is tight, even for two agents \cite{ChristodoulouKST16}.
Can auction formats other than first price do better?
Roughgarden provides a framework for translating communication lower bounds to BPoA bounds \cite{Roughgarden14}. Together with our new lower-bound result, the framework (in particular Theorem VI.1) immediately implies that no auction format with sub-doubly-exponentially many strategies achieves BPoA better than $\nicefrac{1}{2}$. This proves that simultaneous first-price auctions are optimal among this class for all $n$.

\subsection{Main Result and Intuition}

\begin{theorem_main}
	Any (randomized) protocol that guarantees a $\nicefrac{1}{2} + \nicefrac{6}{\log_2(m)}$-approximation to the optimal welfare for two monotone subadditive bidders requires communication $\Omega(e^{\sqrt{m}})$.
\end{theorem_main}

We now provide some intuition for the main steps. The first step in our proof is the construction of a new class of subadditive functions (Section~\ref{sec:construction}). One key feature of our class, if it is to possibly demonstrate hardness better than $\nicefrac{3}{4}$, is that it must not also be fractionally-subadditive (due to Feige's $\nicefrac{3}{4}$-approximation~\cite{Feige09}). Only one general construction exists in prior work (based on Set Cover --- see Section~\ref{sec:construction} for precise description)~\cite{Feige09, BhawalkarR11}.\footnote{It is known that every subadditive function $f$ admits a fractionally-subadditive function that is $\log(m)$-close to $f$ \cite{Dobzinski07, BhawalkarR11}, and this bound is tight by the construction given in \cite{BhawalkarR11}, based on set cover. It is also known that any valuation function for which $v(S) \in \{1,2\}$ for any non-trivial $S$ is also subadditive. But such functions trivially admit a $\nicefrac{3}{4}$-approximation, and therefore don't serve as a useful starting point.} This class is our starting point.

From here, though, we encounter the following barrier: Consider the following line of reasoning, assuming that there exists some set $T$ for which $v_1(T)+v_1(\bar{T}) \geq (1+4\varepsilon)v_1(M)$. Then allocating player $1$ either $T$ or $\bar{T}$ uniformly at random (and the rest to player $2$) guarantees welfare at least $v_1(M)/2 + 2\varepsilon v_1(M) + v_2(T)/2 + v_2(\bar{T})/2 \geq \opt/2 + 2\varepsilon v_1(M)$ by subadditivity of $v_2(\cdot)$. So if $v_1(M) \geq \opt /4$, this allocation guarantees a $(\nicefrac{1}{2}+\varepsilon)$-approximation. If not, then $v_2(M) \geq 3\opt/4$, and awarding all items to player $2$ guarantees a $\nicefrac{3}{4}$-approximation. It is not hard to combine these observations into a simple, deterministic protocol guaranteeing a $(\nicefrac{1}{2}+\varepsilon)$-approximation whenever such a set $T$ exists.

So in order to possibly demonstrate hardness better than $(\nicefrac{1}{2}+\varepsilon)$, our class must further have the property that for all $T$, and all $v(\cdot)$ in our class, $v(T) + v(\bar{T}) \in[v(M), (1+\varepsilon)v(M)]$. That is, $v(\cdot)$ must essentially appear \emph{additive} at the large scale (but may be subadditive at smaller scales).

Indeed, our construction starts with the previous Set-Cover construction and essentially hard-codes that $v(T) + v(\bar{T}) = v(M)$ for all $T$, in a way that maintains subadditivity. We defer further details to Section~\ref{sec:construction}, but do wish to note that this construction itself is likely to be of independent interest for future work, due to the scarcity of known subadditive functions that are ``far from'' fractionally-subadditive. Indeed, our results are the first instance where a non-trivial approximation guarantee is achievable for fractionally-subadditive valuations, but no non-trivial approximation is achievable for subadditive.

From here, we are able to show that our constructions are rich enough to encode \equ.\footnote{In \equ, Alice and Bob are each given $k$-bit strings as input, and are asked to decide whether they are equal or not. \equ\ is known to require deterministic communication $k$ to solve, but admits efficient randomized protocols.} Essentially, the property $v(T) + v(\bar{T}) = v(M)$ is extremely convenient, as it immediately implies that $\opt = v_1(M)$ whenever $v_1(\cdot) = v_2(\cdot)$. As such, our remaining task is to find a doubly-exponentially-large subset $\mathcal{V}$ of our class of valuations for which: (a) $v(M) = \ell$ for all $v \in \mathcal{V}$ and (b) $\opt(v(\cdot),w(\cdot)) \approx 2\ell$ for all $v, w \in \mathcal{V}$. By the convenient property, beating a $\nicefrac{1}{2}$-approximation when both players have valuations from this subset is exactly deciding whether $v_1(\cdot) = v_2(\cdot)$, thereby completing a reduction from \equ. Of course, this only proves our claim for deterministic protocols, as \equ\ is only hard for deterministic communication. We include a complete proof of this reduction in Section~\ref{sec:deterministic} in order to highlight the important aspects of our construction without yet requiring any advanced communication complexity.

Finally, we prove our full lower bound for randomized protocols in Section~\ref{sec:randomized}. Unfortunately, our construction is really an instance of \equ, and is extremely unlikely to admit a reduction from known problems that require exponential randomized communication (such as \disj\ or \ghd). Instead, we propose a new ``near-\equ'' problem (that we call \efs), and directly prove that it requires exponential randomized communication via the information complexity approach of~\cite{Bar-YossefJKS04,Braverman12,BravermanGPW13}. While these tools are now standard in the communication complexity community, they have yet to break into the AGT community. As such, we provide a full exposition in Section~\ref{sec:randomized} (and the associated appendices). Again, we wish to note that Exist-Far-Sets itself may well be of independent interest for future work, especially at the intersection of communication complexity and mechanism design where there is demand for such constructions.

\subsection{Extensions}
Our main result concerns subadditive valuations (for which we prove that a $\nicefrac{1}{2}$-approximation is optimal), which are a proper superclass of fractionally-subadditive (for which a $\nicefrac{3}{4}$-approximation is previously shown to be optimal~\cite{Feige09, DobzinskiNS10}). In Section~\ref{sec:MPH}, we further consider the space between fractionally-subadditive and subadditive valuations via the Maximum-over-Positive-Hypergraphs (MPH) hierarchy~\cite{FeigeFIILS15}. We postpone a formal definition to Section~\ref{sec:MPH}, but note that fractionally-subadditive valuations are equivalent to MPH-$1$, that all monotone functions lie in MPH-$m$, and that all subadditive functions lie in MPH-$m/2$. Our second result is a new protocol for welfare-maximization with two MPH-$k$ bidders:

\begin{theorem_nonum}
There exists a protocol that guarantees a $\nicefrac{1}{2}+\nicefrac{1}{O(\log k)}$ of the optimal welfare for two bidders whose valuations are both subadditive and MPH-$k$ with $\poly(m)$ communication.
\end{theorem_nonum}


In particular, our protocol is an oblivious rounding of the configuration LP.\footnote{That is, while communication is indeed needed to optimally solve the configuration LP, no further communication is necessary in order to round the resulting solution. See~\cite{FeigeFT16} for further discussion on the merits of oblivious versus non-oblivious rounding.} We also wish to note an important corollary of this theorem, when combined with our main result. Our main result proves that a $(\nicefrac{1}{2}+\nicefrac{6}{\log m})$-approximation is impossible with subexponential communication. As all subadditive functions are MPH-$m/2$, this implies that our protocol and lower bound are tight \emph{even up to lower-order terms}.
	
Additionally, as our construction does not admit a $(\nicefrac{1}{2}+\nicefrac{6}{\log m})$-approximation in subexponential communication, it establishes the existence of a constant $C$ such that our constructions are provably not MPH-$Cm$.
This serves as an additional proof for the existence of subadditive functions that lie in high levels of the MPH hierarchy.
The key property we use to claim our guarantee for MPH-$k$ may also be useful for future work to claim lower bounds on the MPH level of specific functions.

\subsection{Related Work}
\label{sec:related}

\paragraph{Communication complexity of combinatorial auctions.} The works most related to ours concern the standard communication complexity of combinatorial auctions. The tables below summarize prior work for various valuation classes. While the $n=2$ table is most relevant for the present paper, the general $n$ table is included for reference. Note that no separate row is needed for hardness of truthful communication, because no such results are known (aside from general communication hardness). 

\begin{table}[h]

\begin{center}
    \begin{tabular}{ || c || c | c | c | c | c ||}
    \hline 
     $n=2$ & Submodular & XOS & Subadditive & General\\ \hline 
    Communication hardness & $\frac{17}{18}$~\cite{DobzinskiV13} & $\frac{3}{4}$~\cite{DobzinskiNS10}& $\frac{3}{4}$~\cite{DobzinskiNS10} & $\frac{1}{2}$~\cite{NisanS06}\\ 
\hline

Communication protocol & $\frac{13}{17}$~\cite{FeigeV10}&$\frac{3}{4}$~\cite{Feige09}&$\frac{1}{2}$~[Trivial] &$\frac{1}{2}$~[Trivial]\\
\hline
Truthful comm. protocol & $\frac{1}{2}$~[Trivial]& $\frac{1}{2}$~[Trivial]& $\frac{1}{2}$~[Trivial]& $\frac{1}{2}$~[Trivial]\\
\hline
    \end{tabular}
\end{center}
\label{table:prior}
\end{table}
\begin{table}[h]
\vspace{-5mm}
\begin{center}
	\resizebox{\textwidth}{!}{
    \begin{tabular}{ || c || c | c | c | c | c ||}
    \hline
     General $n$ & Submodular & XOS & Subadditive & General\\ \hline 
    Comm. hardness & $1-\frac{1}{2e}$~\cite{DobzinskiV13} & $1-(1-\frac{1}{n})^n$~\cite{DobzinskiNS10}& $\frac{1}{2}+\frac{1}{2n}$~\cite{DobzinskiNS10} & $\max\{\frac{1}{n},\frac{1}{\Omega(\sqrt{m})}\}$~\cite{NisanS06}\\ 
\hline

Comm. protocol & $1-\frac{1}{e} + 10^{-5}$~\cite{FeigeV10}&$1-(1-\frac{1}{n})^n$~\cite{Feige09}&$\frac{1}{2}$~\cite{Feige09} & $\max\{\frac{1}{n},\frac{1}{\sqrt{2m}}\}$~\cite{LaviS05}\\ \hline
Truthful comm. &$\frac{1}{O(\sqrt{\log m})}$~\cite{Dobzinski16a} &$\frac{1}{O(\sqrt{\log m})}$~\cite{Dobzinski16a} & $\frac{1}{O(\log m\log\log m)}$~\cite{Dobzinski07}&$\max\{\frac{1}{n},\frac{1}{\sqrt{2m}}\}$~\cite{LaviS05} \\
\hline
    \end{tabular}}
\end{center}
\label{table:prior2}
\end{table}

For context, it is worth noting that all referenced (truthful or not) communication protocols take one of two forms. The first is via solving a particular LP relaxation (called the configuration LP) and rounding the fractional optimum~\cite{FeigeV10, Feige09, LaviS05}. The second is via mechanisms which randomly sample a fraction of bidders to gather statistics, then run a posted-price mechanism on the remaining bidders~\cite{Dobzinski16a, Dobzinski07}. Both classes of mechanisms require bidders to communicate \emph{demand queries}. That is, bidders are asked questions of the form: ``For item prices $p_1,\ldots, p_m$, which set of items maximizes $v_i(S) - \sum_{j \in S} p_j$?'' All of the aforementioned protocols/mechanisms make polynomially many demand queries, and have further polynomial-time overhead. 

Recent work of~\cite{Dobzinski16b} proves a surprising connection between two-player truthful combinatorial auctions, and two-player \emph{simultaneous} (non-truthful) protocols. In particular, any separation between the approximation guarantees achievable by communication-efficient protocols and communication-efficient simultaneous protocols would constitute the first separatation between truthful and non-truthful communication-efficient protocols. Such separations were already known for large $n$~\cite{AlonNRW15, Assadi17}, but not for $n=2$ (and therefore aren't relevant to Dobzinski's framework). As such, the $n=2$ setting is now receiving extra attention, although the desired separation still remains elusive~\cite{BravermanMW18}.

\paragraph{Related results on combinatorial auctions.} As previously referenced, combinatorial auctions are studied via other complexity lenses as well. The most popular alternative is the value-queries model, or standard computational complexity. That is, each bidder is capable only of querying their valuation function on a given set (value query), or has access to the explicit (poly-sized) circuit which computes a value query. In both models, a tight $(1-\nicefrac{1}{e})$-approximation is known for submodular valuations~\cite{Vondrak08, MirrokniSV08, DobzinskiV12b}, and a tight $\Theta(\nicefrac{1}{\sqrt{m}})$-approximation is known for XOS and subadditive valuations~\cite{DobzinskiNS10}. To reconcile these latter impossibility results with the above-referenced positive results, observe that it generally requires exp($m$) value queries (or is NP-hard with explicit circuit access) to compute a demand query. Unlike the communication model, strong separations between guarantees of truthful and non-truthful mechanisms are known in these models~\cite{PapadimitriouSS08, BuchfuhrerDFKMPSSU10, BuchfuhrerSS10, DanielySS15, Dobzinski11, DobzinskiV12a, DobzinskiV12b}. It is also worth noting that some of these approaches also yield communication lower bounds for the restricted class of Maximal-in-Range/VCG-based protocols~\cite{BuchfuhrerSS10, BuchfuhrerDFKMPSSU10, DanielySS15}. For further details of these results, see \cite[Table 1]{DanielySS15}.

\subsection{Summary}\label{sec:summary}
We study the communication complexity of welfare maximization in two-player combinatorial auctions. Our main result establishes that the trivial $\nicefrac{1}{2}$-approximations are in fact optimal among all protocols with subexponential communication. We additionally develop a $(\nicefrac{1}{2}+\nicefrac{1}{O(\log k)})$-approximation whenever both buyers are subadditive and MPH-$k$. Our key innovation is a new class of subadditive functions that are ``far from'' fractionally subadditive, and may be of independent interest for future works. In addition to resolving an open question of~\cite{Feige09, DobzinskiNS10}, our results establish the following corollaries: (a) There is no gap between the approximation ratios achievable by truthful and not-necessarily-truthful mechanisms with $\poly(m)$ communication for two subadditive bidders, (b) For any number of subadditive bidders, simultaneous first price auctions achieve the optimal price of anarchy ($\nicefrac{1}{2}$) among all auctions with sub-doubly-exponentially-many strategies (via~\cite{Roughgarden14}), (c) Our lower bound is tight even up to lower order terms ($\nicefrac{1}{2} + \nicefrac{1}{O(\log m)}$ is achievable in $\poly(m)$ communication, but no better).

\section{Preliminaries}\label{sec:prelim}
We consider the following problem. There is a set $M$ of $m$ items. Alice and Bob each have a valuation function $A(\cdot)$ and $B(\cdot)$, respectively that takes as input subsets of $M$ and outputs an element of $\mathbb{R}_+$. Moreover, $A(\cdot)$ and $B(\cdot)$ are both \emph{monotone} ($v(X \cup Y) \geq v(X)$ for all $X, Y$) and \emph{subadditive} ($v(X \cup Y) \leq v(X) + v(Y)$). Alice and Bob wish to communicate as little as possible about their valuation functions in order to find a welfare-maximizing allocation (that is, the $X$ maximizing $A(X) + B(M\setminus X)$). Formally, we study the following decision problem -- observe that this is a promise problem for which if the input does not satisfy the promise, any output is considered correct.

\begin{definition}[\wm($m,\alpha$)] \wm\ is a communication problem between Alice and Bob:
\begin{itemize}[topsep=1ex,itemsep=0ex,parsep=0ex]
\item \textbf{Alice's Input:} $A(\cdot)$, a monotone subadditive function over $2^{[m]}$; and a target $C$.
\item \textbf{Bob's Input:} $B(\cdot)$, a monotone subadditive function over $2^{[m]}$; and a target $D$.
\item \textbf{Promise:} $C=D$. Also, there either exists an $S\subseteq [m]$ satisfying $A(S) + B(\overline{S}) \geq C$, or for all $S \subseteq [m]$, $A(S) + B(\overline{S}) < \alpha C$.
\item \textbf{Output:} 1 if $\exists S\subseteq [m]$, $A(S)+B(\overline{S}) \geq C$; 0 if $\forall S \subseteq [m]$, $A(S) + B(\overline{S}) < \alpha C$.
\end{itemize}
We will sometimes drop the parameter $m$ when it is irrelevant. We will also refer to any protocol solving $\wm(m,\alpha)$ as an $\alpha$-approximation for $\wm(m)$. 
\end{definition}

Also of interest is the corresponding search problem, which instead asks Alice and Bob to find an $X$ maximizing $A(X) + B(\overline{X})$ (and an $\alpha$-approximation is a protocol guaranteeing a $Y$ satisfying $A(Y) + B(\overline{Y}) \geq \max_{X \subseteq M}\{A(X) + B(\overline{X})\}$). It is easy to see that any $\poly(m)$-communication protocol for the search problem implies a $\poly(m)$-communication protocol for the decision problem (with an extra round of communication). As such, we will prove all lower bounds against the decision problem (as they immediately imply to search as well), and develop all protocols for the search problem (as they immediately imply to decision as well).

\section{Main Construction}\label{sec:construction}
In this section, we present our base construction. In subsequent sections, we show how to leverage this construction to derive our lower bounds. We begin by considering a collection of subsets $\mathcal{S} = \{S_1,\ldots, S_k\}$ where each $S_i \subseteq M$, and defining a useful property that $\mathcal{S}$ may possess. Throughout this section, let $\ell$ denote an even integer $\ge 4$. Some proofs are deferred to Appendix \ref{appx:missing-proofs}.

\begin{defin}[$\ell$-sparse] We say that $\mathcal{S}$ is \emph{$\ell$-sparse} if for all $T_1,\ldots, T_{\ell-1} \in \mathcal{S}$, $\cup_j T_j \neq M$. \end{defin}

That is, $\mathcal{S}$ is $\ell$-sparse if there do \emph{not} exist $\ell-1$ elements of $\mathcal{S}$ such that their union is the entire ground set $M$. We now follow~\cite{Feige09,BhawalkarR11} in defining a class of valuation functions parameterized by a collection of sets. Specifically, let $\mathcal{S} = \{S_1,\ldots, S_k\}$ be an $\ell$-sparse collection. For $X \subseteq M$, define 
\[\sigma_\mathcal{S}(X) :=
\begin{cases}
\min \left\{\abs{Y} : Y \subseteq [k], X \subseteq \bigcup_{i \in Y} S_i \right\},& \text{if }X\text{ is covered by }\mathcal{S};\\
\max\{\ell,k\},& \text{otherwise;}
\end{cases}
\]
where we say ``$X$ is covered by $\mathcal{S}$'' if $X \subseteq \bigcup_{i \in [k]} S_i$.
That is, $\sigma_\mathcal{S}(X)$ is the smallest number of sets from $\mathcal{S}$ whose union contains $X$, or some large number $\max\{\ell,k\}$ if there are no such sets.%
\footnote{Defining $\sigma_\mathcal{S}(X):=\infty$ if $X$ is not covered by $\mathcal{S}$ would have worked as well.}
We can now define our valuation function $f^\ell_\mathcal{S}(\cdot)$:
\begin{enumerate}[label=(\alph*),topsep=1ex,itemsep=0ex,parsep=0ex]
\item \label{fpart1} If $\sigma_\mathcal{S}(X) < \frac{\ell}{2}$, then define $f^\ell_\mathcal{S}(X) = \sigma_\mathcal{S}(X)$ and $f^\ell_\mathcal{S}(\overline{X}) = \ell - \sigma_\mathcal{S}(X)$.
\item\label{fpart2} For any $X$ whose value is not defined in \ref{fpart1}, $f^\ell_\mathcal{S}(X) = \frac{\ell}{2}$.
\end{enumerate}

It is not immediately clear that $f^\ell_\mathcal{S}(\cdot)$ is well-defined; indeed, if $\sigma_\mathcal{S}(X)$ and $\sigma_\mathcal{S}(\overline{X})$ are both $<\frac{\ell}{2}$, then $f^\ell_\mathcal{S}(X)$ is doubly defined. Fortunately, this can never occur when $\mathcal{S}$ is $\ell$-sparse.

\begin{lemma}
\label{lem:well-defined}
If $\mathcal{S}$ is $\ell$-sparse, then $f^\ell_\mathcal{S}(\cdot)$ is well-defined.
\end{lemma}

Now, we would like to prove that $f^\ell_\mathcal{S}$ is monotone and subadditive whenever $\mathcal{S}$ is $\ell$-sparse (Corollary \ref{cor:construction-subadditive}). The following facts about $f^\ell_\mathcal{S}$ and $\sigma_\mathcal{S}$ highlight the key steps in the proof.

\begin{lemma}
\label{fsigmapropslem}
Let $\mathcal{S}$ be $\ell$-sparse. Then:
\leavevmode
\begin{enumerate}[label=(\arabic*),topsep=1ex,itemsep=0ex,parsep=0ex]
\item \label{sigsubadd} $\sigma_\mathcal{S}$ is monotone and subadditive.
\item \label{lminus} For all $X$, $f^\ell_\mathcal{S}(X) = \ell - f^\ell_\mathcal{S}(\overline{X})$.
\item \label{feqsigma} If $\sigma_\mathcal{S}(X) < \frac{\ell}{2}$ or $f^\ell_\mathcal{S}(X) < \frac{\ell}{2}$, then $f^\ell_\mathcal{S}(X) = \sigma_\mathcal{S}(X)$.
\item \label{gtl2} If $f^\ell_\mathcal{S}(X) > \frac{\ell}{2}$ then $f^\ell_\mathcal{S}(X) = \ell - \sigma_\mathcal{S}(\overline{X})$.
\item \label{coverUB} For all $X$, $f^\ell_\mathcal{S}(X)
\le \sigma_\mathcal{S}(X)$.
\end{enumerate}
\end{lemma}


\begin{corollary}
\label{cor:construction-subadditive}
If $\mathcal{S}$ is $\ell$-sparse, then $f^\ell_\mathcal{S}(\cdot)$ is monotone and subadditive.
\end{corollary}

Functions of the form $f^\ell_\mathcal{S}(\cdot)$ will form the basis of our lower bound constructions, which we overview in the following sections.

\section{Deterministic Protocols for Subadditive Valuations}
\label{sec:deterministic}
The construction in Section~\ref{sec:construction} gets us most of the way towards our deterministic lower bound. The remaining step is a reduction from \equ. To briefly remind the reader, Alice receives input $a \in \{0,1\}^k$, and Bob receives input $b \in \{0,1\}^k$. Their goal is to output yes if $a_i = b_i$ for all $i \in [k]$, and no otherwise. It is well-known (see, e.g.,~\cite{KushilevitzN97}) that any deterministic protocol for \equ\ requires communication $\geq k$. 

\begin{theorem} \label{mainthm}
For any even integer $\ell \in [4, \log_2(m)]$, any deterministic communication protocol that guarantees a $(\nicefrac{1}{2} + \nicefrac{1}{\ell})$-approximation to $\wm(m)$ requires communication $\exp\parens{\nicefrac{m}{\ell \cdot 2^\ell}}$. In particular, a guarantee of $\nicefrac{1}{2} + \varepsilon$ requires communication $e^{\nicefrac{\varepsilon m}{2^{\nicefrac{1}{\varepsilon}}}} = e^{\Omega(m)}$, and a guarantee of $\nicefrac{1}{2} + \nicefrac{2}{\log(m)}$ requires communication $e^{m^{\Omega(1)}}$. 
\end{theorem}

Before proceeding with our construction, we'll need one more property of collections of sets:

\begin{defin}[\cite{KleitmanS73}]
	A collection $\mathcal{S} = \{S_1,\ldots, S_k\}$ is \emph{$\ell$-independent} if $\{T_1,\ldots, T_k\}$ is $\ell$-sparse whenever $T_i \in \{S_i, \overline{S_i}\}$. 
\end{defin}

In other words, $\mathcal{S}$ is $\ell$-independent if we can choose either $S_i$ or $\overline{S_i}$ independently, for each $i$, and form an $\ell$-sparse collection no matter our choices. We now proceed with our reduction, which relies on the existence of large $\ell$-independent collections (such collections are known to exist;
at the end of this section we give a precise statement and a proof appears in Appendix \ref{appx:missing-proofs} for completeness). 

\begin{proposition}\label{prop:lindependent}Let $\mathcal{S}$ be an $\ell$-independent collection with $|\mathcal{S}| = k$. Then any deterministic communication protocol that guarantees a $(\nicefrac{1}{2} + \nicefrac{1}{2\ell-3})$-approximation to the optimal welfare for two monotone subadditive bidders requires communication at least $k$. 
\end{proposition}

\begin{proof}
Let $\mathcal{S} = \{S_1,\ldots, S_k\}$ be $\ell$-independent. For each $i$, define $S_i^1:= S_i$, and $S_i^0 := \overline{S_i}$. Now, consider an instance of \equ\ where Alice is given $a$ and Bob is given $b$. Alice will create the valuation function $f^\ell_\mathcal{A}$, where $\mathcal{A}:=\{S_1^{a_1},\ldots, S_k^{a_k}\}$ (i.e. Alice builds $\mathcal{A}$ by taking either $S_i^1$ or $S_i^0$, depending on $a_i$). Bob will create the valuation function $f^\ell_\mathcal{B}$, where $\mathcal{B}:=\{S_1^{b_1},\ldots, S_k^{b_k}\}$. Observe first that $f^\ell_\mathcal{A}(\cdot)$ and $f^\ell_\mathcal{B}(\cdot)$ are indeed well-defined, monotone, and subadditive as $\mathcal{S}$ is $\ell$-independent (and therefore $\mathcal{A}$ and $\mathcal{B}$ are both $\ell$-sparse). 

Observe that if $a = b$, then $\mathcal{A} = \mathcal{B}$ and moreover $f^\ell_\mathcal{A}(\cdot) = f^\ell_\mathcal{B}(\cdot)$. So immediately by part~\ref{lminus} of Lemma~\ref{fsigmapropslem}, the maximum possible total welfare is $\ell$ (indeed, any partition of the items gives welfare $\ell$). 
On the other hand, if there exists an $i$ such that $a_i \neq b_i$ (without loss of generality say that $a_i = 1$ and $b_i = 0$), we claim that welfare $2\ell-2$ is achievable. To see this, consider the allocation which awards $\overline{S_i}$ to Alice and $S_i$ to Bob. Indeed, $f^\ell_\mathcal{A}(S_i) = 1$ (as $S_i \in \mathcal{A}$), so $f^\ell_\mathcal{A}(\overline{S_i}) = \ell - 1$. Similarly, $f_B(\overline{S_i}) = 1$, so $f_B(S_i) = \ell - 1$, achieving total welfare $2(\ell-1)$.\footnote{As an aside, note that welfare exceeding $2\ell-2$ is not possible, as Alice and Bob each value all non-empty sets at least at $1$, and therefore value all strict subsets of $M$ at most at $\ell-1$ by \ref{lminus}. Therefore, the only way Alice or Bob could have value exceeding $\ell-1$ is to get all of $M$, meaning that the other player receives value $0$.} 

So assume for contradiction that a deterministic $\frac{1}{2} + \frac{1}{2\ell-3} > \frac{\ell}{2\ell-2}$-approximation exists to the optimal welfare for $2$ monotone subadditive bidders with communication $< k$. Then such a protocol would solve \equ\ with communication $<k$ by the reduction above, a contradiction.
\end{proof}

Finally, in the next lemma we show how large $k$ can be while guaranteeing an $\ell$-independent collection of size $k$ to exist. This suffices to complete the proof of Theorem~\ref{mainthm}. The lemma is based on a known existential construction using the probabilistic method, which we repeat for the sake of completeness in Appendix \ref{appx:missing-proofs} (explicit constructions of comparable guarantees also exist~\cite{Alon86}).

\begin{lemma}
\label{lem:l-indep-exist}
For all $m$, $x > 1$, and $\ell = \log_2(m) - \log_2(x)$, there exists a $\ell$-independent collection of subsets of $[m]$ of size $k = e^{x/\ell}$.
\end{lemma}

\begin{proof}[Proof of Theorem~\ref{mainthm}]
Combine Proposition~\ref{prop:lindependent} and Lemma~\ref{lem:l-indep-exist}, and the observation that $2\ell-3 \geq \ell$ whenever $\ell \geq 4$. The ``in particular'' parts of the statement follow first by taking $\ell = 1/\varepsilon$ (implying $\ell = \log_2(m) - \log_2(m/2^{1/\varepsilon})$ and $k = e^{\varepsilon m/2^{1/\varepsilon}}$, and then by taking $\ell = \log_2(m)/2$ (implying $\ell = \log_2(m) - \log_2(\sqrt{m})$) and $k = e^{2\sqrt{m}/\log(m)}$. 
\end{proof}

\section{Randomized Protocols for Subadditive Valuations}
\label{sec:randomized}
The construction in Section~\ref{sec:deterministic} carries much of the intuition for our randomized lower bound. However, we clearly cannot reduce from \equ\ and get a randomized lower bound, as \equ\ admits randomized communication-efficient protocols. As such, we will instead directly show that a certain ``near-\equ'' problem requires exponential randomized communication. Our proof uses the information complexity approach popularized in~\cite{Bar-YossefJKS04,Braverman12, BravermanGPW13}, which is now standard in that community. In order to introduce these tools to the AGT community, we will provide a complete exposition starting from the basics. 

Let's first be clear about what a randomized protocol looks like. Alice and Bob have access to a public infinite string of perfectly random bits, $r$. All messages sent by (e.g.) Alice may therefore depend on her input, any messages sent by Bob, and $r$. At the end of the protocol, Alice and Bob will guess yes or no, and the answer should be correct with probability $2/3$.\footnote{As usual, the bound of $2/3$ is arbitrary, as any protocol with success probability $\geq 1/2 + 1/\poly(m)$ can be repeated independently $\poly(m)$ times to achieve a protocol with success probability $2/3$ (and then further repeated to achieve success $1-1/\text{exp}(m)$).} The protocol is only ``charged'' communication for actual messages sent, and not for randomness used. The main result of this section is as follows:

\begin{theorem}[Randomized hardness]
	\label{thm:randomized} Any randomized protocol that guarantees a ($\nicefrac{1}{2}+\nicefrac{6}{\log (m)}$)-approximation to $\wm(m)$ requires communication complexity $\Omega(e^{\sqrt{m}})$. 
\end{theorem}

Let's now understand the issue with our previous construction. Aside from the fact that the previous proof clearly does not extend to a randomized lower bound, the construction itself admits a good randomized algorithm. Specifically, let $\mathcal{S}$ be some $\ell$-independent set, and let exactly one of $\{S, \overline{S}\}$ be in $\mathcal{A}$ for all $S \in \mathcal{S}$ (and also let exactly one of $\{S, \overline{S}\}$ be in $\mathcal{B}$). Let Alice have valuation $f^\ell_{\mathcal{A}}(\cdot)$ and Bob have valuation $f^\ell_{\mathcal{B}}(\cdot)$. The problem is that Alice and Bob are still just trying to determine whether or not $\mathcal{A} = \mathcal{B}$ (that is, if $\mathcal{A} = \mathcal{B}$, then the optimal welfare is at most $\ell$. If not, then the optimal welfare is $2\ell-2$). Since $\mathcal{A}$ and $\mathcal{B}$ are both subsets of $2^M$, the randomized algorithm for \equ\ for inputs of size $2^m$ works.


The natural idea to try next is to reduce from a problem like \disj\ instead (for which randomized protocols indeed require exponential communication). Let's see one natural attempt from our previous construction and why it fails (just for intuition, we will not exhaustively repeat this for all possible reductions). Again let $\mathcal{S}$ denote an $\ell$-independent collection, and again consider any instance $(A, B)$ of \disj\ of size $k$ (recall that $A$ and $B$ are bitstrings of length $k$ and \disj\ asks to decide whether or not there exists an index $i$ with $A_i = B_i = 1$). A first attempt at a reduction might be to let $\mathcal{A}$ contain $S_i$ for all $i$ such that $A_i = 1$, and $\mathcal{B}$ contain all $\overline{S_i}$ such that $B_i = 1$ (but $\mathcal{A}$ will never contain $\overline{S_i}$, and $\mathcal{B}$ will never contain $S_i$). Indeed, with this construction if there exists any index $i$ with $A_i = B_i = 1$, the optimal welfare will be $2\ell-2$ (give Alice $\overline{S_i}$, and Bob $S_i$). Unfortunately, even if there does not exist an index for which $A_i = B_i = 1$, the welfare can still be $\ell-1+\ell/2$. To see this, consider any index $i$ for which $A_i = 1$. Then $f^\ell_{\mathcal{A}}(\overline{S_i}) = \ell-1$. Moreover, as $S_i \notin \mathcal{B}$, and $\mathcal{S}$ is $\ell$-independent, $f^\ell_\mathcal{B}(S_i) = \ell/2$ (because $\mathcal{S}$ is $\ell$-independent, neither $S_i$ nor $\overline{S_i}$ can be covered with fewer than $\ell-1$ of the other sets in $\mathcal{B}$. As such, $f^\ell_{\mathcal{B}}(S_i) = f^\ell_\mathcal{B}(\overline{S_i}) = \ell/2$). So while in the ``yes'' case, the welfare is indeed $2\ell-2$ just like our previous reduction, the welfare in the ``no'' case will be $3\ell/2-1$ as opposed to $\ell$, proving only that a $\nicefrac{3}{4}$-approximation requires exponential randomized communication (which is already known). Of course this is not a formal claim that no reduction from \disj\ is possible, but provides some intuition for why searching for one (or from \ghd, etc.) is likely not the right approach. 

The issue is that our construction is getting much of its mileage from the fact that $v(S) + v(\overline{S}) = \ell$ for all $S$, and reducing from any problem except \equ\ fails to make use of this. So the plan for our new construction is to observe that if $v(\cdot)$ and $w(\cdot)$ are \emph{almost} the same (in a precise sense defined shortly), then we can still claim that $v(S) + w(\overline{S}) \approx \ell$ for all $S$.

The main idea of our construction is as follows: consider still an $\ell$-independent set $\mathcal{S}$. For each $S_i \in \mathcal{S}$, rather than adding either $S_i$ or $\overline{S_i}$ to $\mathcal{A}$, we will add either $S_i \cup \{j\}$ or $\overline{S_i}\cup \{j\}$, where $j$ is a uniformly random element of $M$ (and ditto for $\mathcal{B}$). Adding this random element to each set barely changes the welfare, but makes it significantly harder for Alice and Bob to figure out whether their valuations are nearly identical or not. We now proceed with the construction, followed by a complete proof.

\begin{defin}
Two ordered collections of subsets $\mathcal{X} = \langle X_1, \dots, X_k\rangle; \mathcal{Y} = \langle Y_1, \dots, Y_k\rangle$ of $M$ are \emph{$\ell$-compatible} if
\begin{enumerate}[label=(\arabic*),topsep=1ex,itemsep=0ex,parsep=0ex]
\item \label{setsize} $|X_i| = |Y_i| =\frac{m}{2} + 1$ for all $i$.
\item \label{similarordifferent} Either $\abs{X_i \triangle Y_i} = 2$ or $\abs{X_i \cap Y_i} = 2$ for all $i$.%
\footnote{$\triangle$ represents symmetric difference, $|X_i \cup Y_i| - |X_i \cap Y_i|$.}
\item \label{propertyp} $X_1, \dots, X_k$ are $\ell$-sparse, as are $Y_1, \dots, Y_k$.
\item \label{containsmallsets} For any subset $S \subseteq M$ of size less than $\frac{\ell}{2}$, at least one of $X_1, \dots, X_k$ contains $S$, as does at least one of $Y_1, \dots, Y_k$.
\end{enumerate}
\end{defin}

The main idea is as follows: for any $\ell$-compatible $\mathcal{X},\mathcal{Y}$, consider the valuation functions $f^\ell_\mathcal{X}(\cdot)$ and $f^\ell_\mathcal{Y}(\cdot)$. If for some $i$, $|X_i \cap Y_i| = 2$, this roughly corresponds to the ``not equal'' case in the previous construction, and welfare near $2\ell$ is achievable. If instead, for all $i$ $|X_i \triangle Y_i| = 2$, this roughly corresponds to the ``equal'' case in the previous construction, and welfare near $\ell$ is the best achievable. We first state this formally, and then follow with a proof that randomized protocols require exponential communication to distinguish these two cases.
\begin{lemma}\label{lem:welfaretoefs}
Let $\mathcal{X},\mathcal{Y}$ be $\ell$-compatible. If for some $i$, $\abs{X_i \cap Y_i} = 2$, then welfare $2(\ell - 1)$ is achievable between $f^\ell_{\mathcal{X}}(\cdot)$ and $f^\ell_{\mathcal{Y}}(\cdot)$. Otherwise, the maximum achievable welfare is at most $\ell+1$.
\end{lemma}

\begin{proof}
Suppose that for some $i$, $\abs{X_i \cap Y_i} = 2$. Consider the allocation that awards $Y_i$ to Alice and $\overline{Y_i}$ to Bob. Then Bob clearly has value $\ell-1$, as $f_\mathcal{Y}^\ell(\overline{Y_i}) = \ell-1$. Also, as $|Y_i| = |X_i|= m/2+1$ and $|X_i \cap Y_i|=2$, we necessarily have $X_i \cup Y_i = M$. This implies that $Y_i \supseteq \overline{X_i}$, and therefore $v_\mathcal{X}^\ell(Y_i) \geq v_\mathcal{X}(\overline{X_i}) = \ell-1$. So welfare $2\ell-2$ is achievable (and again, optimal, as no bidder can achieve value $\ell$ without receiving all of $M$). 

Now suppose that for all $i$, $\abs{X_i \triangle Y_i} = 2$, and further suppose for contradiction that total welfare $>\ell + 1$ is achievable, by giving $S$ to Alice and $\overline{S}$ to Bob. Then one of the players (without loss of generality, say it is Bob) has value $>\frac{\ell}{2}$, so it must have been the case that $f_\mathcal{Y}^\ell(\overline{S})$ was defined to be $\ell - \sigma_\mathcal{Y}(S)$, and $\sigma_\mathcal{Y}(S)< \frac{\ell}{2}$. 

Now, observe that because $|X_i \triangle Y_i| = 2$ for all $i$ that $\sigma_\mathcal{X}(S) \le \sigma_\mathcal{Y}(S) + 1$. Indeed, let $S = Y_{i_1} \cup \dots \cup Y_{i_{\sigma_\mathcal{Y}(S)}}$. Then there is exactly one element in $Y_{i_j}$ that is not also in $X_{i_j}$, and we therefore conclude that $|S \setminus (X_{i_1} \cup \dots \cup X_{i_{\sigma_\mathcal{Y}(S)}})| \leq \sigma_\mathcal{Y}(S) < \ell/2$. By criterion~\ref{containsmallsets} of $\ell$-compatibility, there is some $X_i$ that contains all of these elements, and so $S \subseteq X_i \cup X_{i_1} \cup \dots \cup X_{i_{\sigma_B(S)}}$, witnessing that $\sigma_\mathcal{X}(S) \le \sigma_\mathcal{Y}(S) + 1$. 

Finally, if $\sigma_\mathcal{X}(S) < \frac{\ell}{2}$ then $f^\ell_\mathcal{X}(S) = \sigma_\mathcal{X}(S) \le \sigma_\mathcal{Y}(S) + 1$, so the total welfare is at most $\ell + 1$, a contradiction. Otherwise, $\sigma_\mathcal{X}(S) = \frac{\ell}{2}$, which we claim implies that $f_\mathcal{X}^\ell(S) \leq \ell/2$. Indeed, if $f_\mathcal{X}^\ell(S) > \ell/2$, then it is because $\sigma_\mathcal{X}(\overline{S}) < \ell/2$. But then we would have $\sigma_\mathcal{X}(S) + \sigma_\mathcal{X}(\overline{S}) < \ell$, implying a cover of $M$ with $< \ell$ and contradicting $\ell$-sparsity of $\mathcal{X}$. Observe that in both cases we may conclude that $f_\mathcal{X}^\ell(S) \leq \sigma_\mathcal{X}^\ell(S)$. 

Now we may conclude that the total welfare is $f_\mathcal{X}^\ell(S) + f^\ell_\mathcal{Y}(\overline{S}) \leq \sigma_\mathcal{X}(S) + \ell-\sigma_\mathcal{Y}(S) \leq \ell+1$, again a contradiction. We have now reached a contradiction from all branches starting from the assumption that welfare $> \ell+1$ is achievable, and may now conclude that the maximum possible welfare is indeed $\leq \ell + 1$, as desired.
\end{proof}

The remaining step is now ``simply'' to show that distinguishing between the two cases requires exponential randomized communication. 

\subsection{\fs\ and \efs}
Towards proving our lower bound, we'll define the following two problems, which may themselves be of independent interest, at least within the combinatorial auctions community, as a ``near-\equ'' problem which requires exponential randomized communication. Below, note that both \fs\ and \efs\ are promise problems: if the input doesn't satisfy any of the stated conditions, arbitrary output is considered correct. 

\begin{defin}[\textsc{Far-Sets}($m$)] \fs\ is a communication problem between Alice and Bob: 
\begin{itemize}[topsep=1ex,itemsep=0ex,parsep=0ex]
\item \textbf{Alice's Input:} $X \subseteq M$, with $|X| = m/2+1$. 
\item \textbf{Bob's Input:} $Y \subseteq M$, with $|Y| = m/2 + 1$.
\item \textbf{Promise:} Either $\abs{X \triangle Y} = 2$ or $\abs{X \cap Y} = 2$.
\item\textbf{Output:} $0$ if $\abs{X \triangle Y} = 2$; 1 if $\abs{X \cap Y} = 2$.
\end{itemize}
\end{defin}

\begin{defin}[\textsc{Exist-Far-Sets}($m,k,\ell$)] \efs\ is a communication problem between Alice and Bob:
\begin{itemize}[topsep=1ex,itemsep=0ex,parsep=0ex]
\item \textbf{Alice's Input:} $\mathcal{X}= \langle X_1,\ldots, X_k\rangle$. Each $X_i \subseteq M$. 
\item \textbf{Bob's Input:} $\mathcal{Y}= \langle Y_1,\ldots, Y_k\rangle$. Each $Y_i \subseteq M$.
\item \textbf{Promise:} $\mathcal{X}$ and $\mathcal{Y}$ are $\ell$-compatible.
\item\textbf{Output:} $\bigvee_{i \in [k]} \fs(X_i, Y_i)$. Observe that if the \efs\ promise above is satisfied, then by definition the \fs\ promise is satisfied for all $(X_i, Y_i)$ (but not necessarily vice versa --- the \efs\ promise is strictly stronger due to $\ell$-sparsity).
\end{itemize}
\end{defin}

Observe that \efs\ is \emph{exactly} the problem referenced in Lemma~\ref{lem:welfaretoefs} (we state this formally below). Therefore, the goal of this section is to lower bound the randomized communication complexity of \efs. 

\begin{corollary}\label{cor:efswelfare}Let $C(m,k,\ell)$ be such that every randomized communication protocol which solves any given instance of \efs($m,k,\ell$) with probability at least $2/3$ has communication complexity at least $C(m,k,\ell)$. Then every randomized communication protocol which solves any given instance of $\wm(m,\frac{1}{2}+\frac{1}{\ell-1})$ with probability at least $2/3$ has communication complexity at least $C(m,k,\ell)$.
\end{corollary}

Our plan of attack will look very similar to Braverman's lower bound on the randomized communication complexity of \disj~\cite{Braverman12}.

\subsection{Information Theory Preliminaries}
\label{sec:infotheory}

Here, we provide some basic facts about information theory and information complexity. These are the standard preliminaries one would find in a paper on information complexity (e.g.~\cite{BravermanGPW13}). Below, when we refer to a distribution $\mu$, we use $\mu(\omega)$ to denote the probability that $\omega$ is sampled from $\mu$. All logarithms taken in this section are base-2. Also for this section, all distributions and random variables are supported on a finite set $\Omega$. If $\mu(\omega) = 0$ for some $\omega \in \Omega$, we let $0\cdot \log_2(1/0):= \lim_{x \rightarrow 0} x \cdot \log_2(1/x) = 0$. 

\begin{definition}[Entropy] Let $\mu$ be a probability distribution over a finite set $\Omega$. The (Shannon) entropy of $\mu$, denoted by $H(\mu)$, is defined as $H(\mu):= \sum_{\omega \in \Omega} \mu(\omega) \log(\frac{1}{\mu(\omega)})$. If $A$ is a random variable distributed according to $\mu_A$, we also write $H(A) := H(\mu_A)$.
\end{definition}

\begin{definition}[Conditional Entropy] Let $A$ and $B$ be two random variables supported on a finite set $\Omega$. Then the conditional entropy of $A$, conditioned on $B$ is $H(A|B):= \sum_{b \in \Omega} \Pr[B = b] \cdot H(A|B = b)$.\footnote{To be clear, by $H(A|B=b)$ we mean the entropy of the random variable $A$, when drawn conditioned on the event $B= b$.}
\end{definition}

Observe that as $H(\cdot)$ is a strictly concave function, $H(A|B) \leq H(A)$ for all $B$ (with equality iff $A$ and $B$ are independent). 

\begin{fact}\label{fact:sumentropy} $H(A,B) = H(A) + H(B|A)$. Here, $H(A,B)$ denotes the entropy of the random variable $(A,B)$. 
\end{fact}

Fact~\ref{fact:sumentropy} above intuitively says that the entropy of a tuple of random variables is equal to the entropy of the first, plus the entropy of the second conditioned on the first. Note that if $A$ and $B$ are independent, then the joint entropy $H(A,B) = H(A) + H(B)$.

\begin{definition}[Mutual Information] For two random variables $A, B$, the mutual information between $A$ and $B$, denoted by $I(A;B)$ is: $I(A;B) := H(A) - H(A|B) = H(B) - H(B|A)$.
\end{definition}

\begin{definition}[Conditional Mutual Information] For three random variables $A, B, C$, the mutual information between $A$ and $B$, conditioned on $C$ is denoted by $I(A;B|C)$, and $I(A;B|C):= \sum_{c \in \Omega} \Pr[C = c] \cdot I(A|C=c;B|C=c)$.
\end{definition}

\begin{fact}[Chain Rule]\label{fact:chainrule} Let $A,B,C,D$ be random variables. Then $I(A,B;C|D) = I(A;C|D)+I(B;C|A,D)$.
\end{fact}

Fact~\ref{fact:chainrule} above intuitively says that the information learned about $(A,B)$ from $C$ (conditioned on $D$) can be broken down into the information learned about $A$ from $C$ (conditioned on $D$), plus the information learned about $B$ from $C$ (now conditioned on $A$ in addition to $D$). 

\begin{definition}[KL Divergence]\label{def:kl} We denote by $\mathbb{D}(A||B)$ the Kullback-Leibler Divergence between $A$ and $B$, which is defined as $\mathbb{D}(A||B) := \sum_{\omega \in \Omega} \Pr[A =\omega] \cdot \log \frac{\Pr[A = \omega]}{\Pr[B = \omega]}$.
\end{definition}

\begin{fact}\label{fact:KL} For any random variables $A, B, C$, $I(A; B|C) = \mathbb{E}_{b,c}[\mathbb{D}(A_{bc}||A_c)]$. Here, $A_{bc}$ denotes the random variable $A$ conditioned on $B=b, C=c$, and $A_c$ denotes the random variable $A$ conditioned on $C = c$.
\end{fact}

\begin{fact}[Pinsker's Inequality]\label{fact:PI} For any pair of random variables $A,B$ of finite support, $||P-Q||_1 \leq \sqrt{2 \mathbb{D}(P||Q)}$. Here, $||P-Q||_1:= \sum_{\omega \in \Omega} |\Pr[P = \omega] - \Pr[Q = \omega]|$.

\end{fact}

\begin{definition}[Information Complexity] The (internal) Information Complexity of a communication protocol $\pi$ with respect to a distribution $\mu$ over pairs $(X,Y)$ of inputs is defined as follows. Let $\Pi(X,Y)$ denote the random variable which is the transcript produced when Alice and Bob participate in protocol $\pi$ with inputs $(X,Y)$, when $(X,Y)$ are drawn from $\mu$. Then $IC_\mu(\pi) := I(\Pi(X,Y);Y|X) + I(\Pi(X,Y);X|Y)$. 
\end{definition}

Above, the ``transcript'' refers to all communication between Alice and Bob (including the order bits were sent, who sent them, etc., \emph{and including any public randomness}) when participating in protocol $\pi$. In particular, it is always possible to glean the output produced by $\pi$ from viewing the transcript (but possibly additional information). Informally, the Information Complexity captures the amount of information Alice learns about Bob's input from participating in $\pi$ (given that she already knows her own input, the public randomness, and that their joint input is drawn from $\mu$), plus what Bob learns about Alice's input. Intuitively, it should be impossible for a protocol to convey $C$ bits of information without $C$ bits of communication. Indeed, this is the case:

\begin{lemma}[\cite{BravermanR11}]\label{lem:br} For any distribution $\mu$ and protocol $\pi$, $IC_\mu(\pi) \leq CC(\pi)$ (where $CC(\pi)$ denotes the worst-case number of bits communicated during protocol $\pi$ on any input). 
\end{lemma}

We conclude with a few more basic facts about communication protocols. Lemma~\ref{lem:protocol} below captures one key difference between communication protocols and algorithms with access to the entire input. Lemma~\ref{lem:protocol} below refers to \emph{private randomness}, which are random bits which are known only to Alice (but not Bob) and vice versa. Such bits are also not counted towards the communication cost of the protocol (unless Alice wishes to share her private randomness with Bob). 

\begin{lemma}\label{lem:protocol}
Let $P_\pi(\cdot, \cdot, \cdot)$ be a function where $P_\pi(z, X, Y)$ denotes the probability that Alice with input $X$ and Bob with input $Y$ produce transcript $z$ when participating in a protocol $\pi$ \emph{over the randomness of any private randomness used} (as public randomness is already accounted for in the transcript). Then there exist functions $Q_\pi(\cdot, \cdot)$ and $R_\pi(\cdot, \cdot)$ such that $P_\pi(z, X, Y) = Q_\pi(z, X) \cdot R_\pi(z, Y)$.
\end{lemma}

The proof of Lemma~\ref{lem:protocol} is straightforward. Essentially, $Q_\pi(z, X)$ is the probability that Alice \emph{doesn't} deviate from transcript $z$ with input $X$, conditioned on Bob communicating according to transcript $z$ so far. Similarly, $R_\pi(z, Y)$ denotes the probability that Bob doesn't deviate from transcript $z$ with input $Y$, conditioned on Alice communicating according to transcript $z$ so far. These probabilities are well-defined because Alice must choose her future messages based only on the transcript so far (including the public randomness) and her input $X$ (and Bob's must be only on the messages so far and $Y$), as well as her private randomness. Once confirming that these probabilities are well-defined, it is easy to see that indeed $P_\pi(z, X, Y) = Q_\pi(z,X) \cdot R_\pi(z,Y)$.

Finally, the lemma below states that lower bounds on the information complexity of any protocol that \emph{only} uses private randomness also lower bound the information complexity of any protocol which uses public randomness. This initially may seem counterintuitive, since the opposite is true for communication. Both simple claims below have ``approximate'' versions in the other directions (discussed in the cited references), but we only use the easy directions.

\begin{lemma}[Folklore, see~\cite{Newman91,BravermanG14,BrodyBKLSV16}]\label{publicrandlemma}
Let $\pi$ be a protocol, and $\mu$ be a distribution over inputs. Then:
\begin{itemize}
\item If $\pi$ uses private randomness, there exists a protocol $\pi'$ using public randomness with exactly the same output as $\pi$, and $CC(\pi') \leq CC(\pi)$. But maybe $IC_\mu(\pi') \gg IC_\mu(\pi)$.
\item If $\pi$ uses public randomness, there exists a protocol $\pi'$ using private randomness with exactly the same output as $\pi$, and $IC_\mu(\pi') \leq IC_\mu(\pi)$. But maybe $CC(\pi') \gg CC(\pi)$. 
\end{itemize}
\end{lemma}

Both claims in Lemma~\ref{publicrandlemma} follow by simple reductions. For the first bullet, simply use all odd bits of the public randomness string as private randomness for Alice, and all even bits of the public randomness string as private randomness for Bob. Then the output of the protocol is identical, and the communication has not changed. However, Bob now \emph{knows} Alice's private randomness, so the protcol may reveal significantly more information than previously (one example to have in mind is that perhaps the protocol has Alice output one uniformly random bit of her input. With private randomness, Bob learns very little about Alice's input upon seeing the bit. With public randomness, Bob learns exactly one bit of Alice's input). For the second bullet, simply use Alice's private randomness as the public randomness. That is, whenever the protocol requests random bits, Alice outputs these bits from her private random string. These bits are completely independent of her input, and therefore reveal no additional information. However, the communication might become enormous, as the randomness is now being directly communicated, and counts towards the communication cost. To use Lemma~\ref{publicrandlemma}, our lower bounds proceed by first lower bounding the information complexity of \efs\ with private randomness, using Lemma~\ref{publicrandlemma} to lower bound the information complexity of \efs\ with public randomness, using Lemma~\ref{lem:br} to lower bound the communication complexity of \efs\ with public randomness (if desired, we could then use Lemma~\ref{publicrandlemma} to further lower bound the communication complexity of \efs\ with private randomness). The point is just that exponential communication is required with either public or private randomness.

\subsection{From \efs\ to \fs}
In this section, we show how to lower bound the randomized communication complexity of \efs, provided we have \emph{any} lower bound on the information complexity of (certain instances of) \fs. This section is analogous to Section 7.2 of~\cite{Braverman12}, but we repeat the approach here in order to properly introduce these ideas to the AGT community.\footnote{Also, because of the additional promise of \efs, we are unfortunately unaware of prior work which we can cite as a black-box. But the proof below really follows from exactly the same ideas as~\cite{Braverman12} after accounting for the promise.} To get started, we need some additional notation for promise problems.

\begin{definition}[Promise problem] Let $f$ be some function mapping $\{0,1\}^m \times \{0,1\}^m \rightarrow \{0,1\}$, and let $P \subseteq \{0,1\}^m \times \{0,1\}^m$. Then the communication problem \emph{solving $f$ under promise $P$} refers to the communication problem which requires Alice with input $A$ and Bob with input $B$ to output $f(A,B)$ whenever $(A,B) \in P$ (and they may provide arbitrary output otherwise). 
\end{definition}

\begin{definition}[$b$-compatible inputs] Say that an input $(X,Y)$ is $b$-compatible with respect to $f$ and $P$ if $(X,Y) \in P$ and $f(X,Y) = b$.
\end{definition}

\begin{definition}[$(k,z)$-safe distributions] Say that a distribution $\mu$ over $\{0,1\}^m \times \{0,1\}^m$ and promise $P\subseteq\{0,1\}^m \times \{0,1\}^m$ are \emph{$(k,z)$-safe with respect to promise $P^*\subseteq (\{0,1\}^m \times \{0,1\}^m)^k$} if for all $j \in [k]$, and all $(X_j,Y_j) \in P$, the probability that $\langle (X_1,Y_1),\ldots, (X_k, Y_k)\rangle \in P^*$ is at least $z$, when $(X_i,Y_i)$ are drawn i.i.d. from $\mu$ for all $i \neq j$.
\end{definition}

Below is the main tool from information complexity that we'll use, which provides a method for proving communication lower bounds for $k$-dimensional communication problems via information complexity lower bounds on a related $1$-dimensional problem (plus some technical assumptions to handle the promises).

\begin{theorem}[follows from \cite{Braverman12}]\label{thm:braverman} Let $f$ be some function mapping $\{0,1\}^m \times \{0,1\}^m \rightarrow \{0,1\}$, and let $P \subseteq \{0,1\}^m \times \{0,1\}^m$ be some promised set of inputs. Let $F^*$ be defined so that 
$$
F^*((X_1,\ldots, X_k),(Y_1,\ldots, Y_k)) = \bigvee_{i \in k} f(X_i, Y_i).
$$
Let also $P^* \subseteq (\{0,1\}^m \times \{0,1\}^m)^k$ be some promised set of inputs such that $P^* \subseteq P^k$ (that is, every coordinate of an element of $P^*$ is in $P$). 

Let also $\mu$ be any distribution over inputs that are $0$-compatible with respect to $f$ and $P$. Let also $\mu$ and $P$ be $(k,z)$-safe with respect to $P^*$.

Then, if any protocol $\pi$ that solves $f$ under promise $P$ with probability at least $q$ has $IC_\mu(\pi) \geq c$, any randomized protocol $\pi^*$ that solves $F^*$ under promise $P^*$ with probability at least $q/z$ has $CC(\pi^*) \geq kc$.
\end{theorem}
\begin{proof}
We use exactly the same reduction proposed in Section 7.2 of~\cite{Braverman12} for \disj. The analysis just requires a little extra work to accommodate the promises. Assume for contradiction that there exists a randomized protocol $\pi^*$ that solves $F^*$ under promise $P^*$ with probability at least $q/z$ and $CC(\pi^*)<kc$. We prove the contrapositive of the theorem statement by providing a randomized protocol $\pi$ that solves $f$ under promise $P$ with probability at least $q$, with $IC_\mu(\pi) < c$. Here is the reduction (recall that $\pi^*$ is the assumed protocol for $F^*$ under promise $P^*$, and we will use $\pi$ to refer to the designed protocol): 

\begin{itemize}
\setlength{\itemsep}{0pt}
\item Alice and Bob are given input $(X, Y) \in P$.
\item Alice and Bob use the shared randomness to select an $i \in [k]$ uniformly at random. Let $X_i:= X$ and $Y_i:= Y$.
\item Alice and Bob use shared randomness to publicly sample $X_1,\ldots, X_{i-1}$ i.i.d.~from $\mu_X$, and $Y_{i+1}, \ldots, Y_k$ i.i.d.~from $\mu_Y$. Here, $\mu_X$ denotes the marginal of $\mu$ restricted to $X$ (ditto for $\mu_Y$). 
\item Alice privately samples $X_j$ from $\mu_X|Y_j$ for all $j > i$. That is, Alice samples $X_j$ from $\mu_X$, conditioned on $Y_j$. Bob privately samples $Y_j$ from $\mu_Y|X_j$ for all $j < i$.
\item Alice and Bob run protocol $\pi^*$ on input $(X_1, \dots, X_k; Y_1, \dots, Y_k)$ (refer to this as $(\mathcal{X},\mathcal{Y})$ for notational simplicity) and output the answer.
\end{itemize}

Let's first observe the following about the correctness of the above protocol: First, maybe $(\mathcal{X},\mathcal{Y}) \notin P^*$. In this case, we have no guarantees about the success of the protocol (because $\pi^*$ may behave arbitrarily). However, because $(\mu, P)$ are $(k,z)$-safe with respect to $P^*$, we know that for all $(X,Y)$, the resulting $(\mathcal{X},\mathcal{Y})$ is in $P^*$ with probability at least $z$. 

Moreover, observe that, conditioned on $(\mathcal{X}, \mathcal{Y}) \in P^*$, protocol $\pi^*$ is correct with probability at least $q/z$ by definition. Finally, we observe that because $\mu$ is $0$-compatible with respect to $f$ and $P$, $f(X_j, Y_j) = 0$ for all $j \neq i$. Therefore, $f(X,Y) = \bigvee_{j=1}^k f(X_j, Y_j) =  F^*(\mathcal{X},\mathcal{Y}) $. So whenever protocol $\pi^*$ is correct on $(\mathcal{X},\mathcal{Y})$, our protocol $\pi$ is correct on $(X,Y)$. \textbf{At this point, we may conclude that our protocol is correct with probability at least $z \cdot q/z = q$.} The remaining step is to compute its information complexity with respect to $\mu$. 

First, we wish to point out that indeed the \emph{communication complexity} of $\pi$ can be quite high. Indeed, it is blowing up what should be a one-dimensional problem into a $k$-dimensional problem before solving it. Intuitvely, though, we'd like to say that only a $1/k$-fraction of this communication is actually being used to solve our original instance. It's hard to make this a formal statement, but this intuition can be made formal if we look at information complexity instead: In expectation, only a $1/k$-fraction of the \emph{information our protocol learns} is relevant to the original instance. 

Below, recall that $I(P; Q)$ denotes the mutual information of random variables $P$ and $Q$ and $I(P; Q \mid R)$ denotes the expected mutual information of $P$ and $Q$ conditioned on $R$. We will use $\Pi(X,Y)$ to denote the random variable corresponding to the transcript of the protocol $\pi$, when $(X,Y)$ are drawn from $\mu$, and will use $\Pi^*(\mathcal{X},\mathcal{Y})$ to denote the random variable corresponding to the transcript of $\pi^*$ when inputs $(\mathcal{X},\mathcal{Y})$ are drawn from $\mu^k$. Let's begin by writing the information the information that Alice learns about $Y$ given $X$ from our protocol $\pi$:
\begin{align*}
I(\Pi(X,Y); Y \mid X) &= I(\Pi^*(\mathcal{X},\mathcal{Y}); Y \mid i, X_1, \dots, X_{i - 1}, X, X_{i+ 1}, \dots, X_k, Y_{i + 1}, \ldots, Y_k).
\end{align*}

Let's parse the above statement before proceeding. It is essentially saying that what Alice learns about $Y$ from $\Pi(X,Y)$ given that she already knows $X$ is exactly the same as what she learns from the random variable $\Pi(\mathcal{X},\mathcal{Y})$ (because these random variables are identical), except now Alice already knows $\mathcal{X}$ and $Y_{i+1},\ldots, Y_k$ (because they were sampled publicly). Now we can perform some manipulations based on the facts from Section~\ref{sec:infotheory}. For the first step, we'll separate out the conditioning on $i$, as $i$ is independent of all other random variables (we'll also start replacing $X_1,\ldots, X_{i-1},X, X_{i+1},\ldots, X_k$ with $\mathcal{X}$ to save space):
\begin{align*}
I(\Pi^*(\mathcal{X},\mathcal{Y}); Y|i, \mathcal{X}, Y_{i+1},\ldots,  Y_k) &= \sum_{i=1}^k \frac{1}{k} \cdot I(\Pi^*(\mathcal{X},\mathcal{Y}); Y|\mathcal{X}, Y_{i+1}, \ldots, Y_k)\\
&=  \frac{1}{k} \cdot \sum_{i=1}^kI(\Pi^*(\mathcal{X},\mathcal{Y}); Y_i| \mathcal{X}, Y_{i+1}, \ldots, Y_k).
\end{align*}

Above, the second equality is simply relabeling $(X,Y)$ as $(X_i,Y_i)$, as they are identically distributed and independent of all other random variables. From here, we can repeatedly apply the chain rule. Specifically, recall that the chain rule implies that $$
I(\Pi^*(\mathcal{X},\mathcal{Y});Y_k|\mathcal{X})+I(\Pi^*(\mathcal{X},\mathcal{Y});Y_{k-1}|\mathcal{X}, Y_k) = I(\Pi^*(\mathcal{X},\mathcal{Y});Y_{k-1},Y_k|\mathcal{X}).
$$
More generally, for any $i$, 
$$
I(\Pi^*(\mathcal{X},\mathcal{Y});Y_i,\ldots, Y_k|\mathcal{X}) + I(\Pi^*(\mathcal{X},\mathcal{Y});Y_{i-1}|\mathcal{X}, Y_i, \ldots, Y_k) = I(\Pi^*(\mathcal{X},\mathcal{Y});Y_{i-1},\ldots, Y_k|\mathcal{X}).
$$ 
As such, we get that:
\begin{align*}
 \frac{1}{k}\cdot  \sum_{i = 1}^k I(\Pi^*(\mathcal{X},\mathcal{Y}); Y_i \mid \mathcal{X}, Y_{i + 1}, \dots, Y_k) = \frac{1}{k} \cdot I(\Pi^*(\mathcal{X},\mathcal{Y}); \mathcal{Y} \mid \mathcal{X}).
\end{align*}
We may now conclude that $I(\Pi;Y|X) = \frac{1}{k} \cdot I(\Pi^*(\mathcal{X},\mathcal{Y});\mathcal{Y}|\mathcal{X})$. The exact same argument swapping the roles of $X$ and $Y$ yields that $I(\Pi(X,Y); X \mid Y) = \frac{1}{k} \cdot I(\Pi^*(\mathcal{X},\mathcal{Y}); \mathcal{X} \mid \mathcal{Y})$. These are the key claims: even though the communication of the protocol $\pi$ may be huge, the information complexity is small. In a formal sense, only a $1/k$ fraction of the information conveyed through throughout protocol $\pi$ is conveyed about the specific indices where we placed $(X,Y)$. From here, we now conclude:
\begin{align*}
IC_\mu(\pi) &= I(\Pi(X,Y); X \mid Y) + I(\Pi(X,Y); Y \mid X) = \frac{1}{k}(I(\Pi^*(\mathcal{X},\mathcal{Y}); \mathcal{X} \mid \mathcal{Y}) + I(\Pi^*(\mathcal{X},\mathcal{Y}); \mathcal{Y} \mid \mathcal{X}))\\
&= \frac{1}{k} IC_{\mu^k}(\pi^*) \le \frac{1}{k} CC(\Pi).
\end{align*}
Recall that we assumed for for contradiction that $CC(\Pi) <ck$, and therefore we would conclude that $IC_\mu(\pi) <c$. The contrapositive proves the theorem statement.
\end{proof}

\subsection{The Information Complexity of \textsc{Far-Sets}}
Theorem~\ref{thm:braverman} reduces our search for a communication lower bound on \efs\ to a search for an information complexity lower bound on \fs. In this section, we prove such a bound. To get started, let's first define the distribution $\mu$ for which we'll prove our information complexity lower bound:

\begin{defin}
Let $\mu$ denote the uniform distribution over all pairs of sets $(X, Y)$ of size $\frac{m}{2} + 1$ such that $\abs{X \triangle Y} = 2$.
\end{defin}

Recall that intuitively, the goal of this section is to prove that any protocol that solves \fs\ correctly with probability bounded away from $1/2$ \emph{on all instances} must result in Alice learning \emph{some} information about $Y$ (or Bob learning some information about $X$) when this protocol is run on instances $(X,Y)$ drawn from $\mu$. Note that it is crucial to assume that the protocol is correct on \emph{all} instances, and not just those drawn from $\mu$ (as all instances drawn from $\mu$ can be correctly answered by communicating nothing and outputting 0).

\begin{theorem} \label{biglemma}
Let $\pi$ be a randomized protocol (with public or private randomness) that solves \textsc{Far-Sets} correctly with probability greater than $0.8$ on every input (that satisfies the promise). Then $IC_\mu(\pi) > \frac{1}{4m^5}$.
\end{theorem}

\begin{proof}
For simplicity of notation, assume that $m$ is a multiple of $4$ (if not, it is just an issue of being more careful with indices). For further simplicity of notation, let $n = m/2$. The main idea of the proof is to derive a contradiction from the following two arguments. First, if $IC_\mu(\pi)$ is negligible, then there exists a ``chain'' of sets $S_1,\ldots, S_n$ such that $\abs{S_1 \cap S_n} = 2$; and for all $i$, $\abs{S_i \triangle S_{i + 1}} = 2$, Alice with input $S_i$ must not be able to effectively distinguish between when Bob has $S_{i+1}$ or $S_{i-1}$ just from the transcript of $\pi$. 

On the other hand, if $\pi$ solves \fs\ on \emph{every} instance (that satisfies the promise), then when Alice has input $S_1$, Alice must be able to effectively distinguish between when Bob has $S_2$ and $S_n$ just from the transcript of $\pi$ (because the correct output is different in these two cases, and the distribution of transcripts must therefore be noticeably different). Due to the nature of communication protocols, these two claims will turn out to be contradictory. These are the key steps in the proof approach; we now begin formally. In the lemma statement below, for two random variables $X, Y$ with finite support we let $||X-Y||_1 := \sum_z |\Pr[X = z] - \Pr[Y = z]|$. 

\begin{proposition}\label{prop:correctness} Let $\pi$ be any private-randomness protocol that correctly solves \fs\ on every instance (satisfying the promise) with probability at least $0.8$. Then no chain of sets $S_1,\ldots, S_n$ satisfies the following properties simultaneously:

\begin{itemize}
\setlength{\itemsep}{0pt}
\item $|S_i| = m/2+1$ for all $i$.
\item $|S_1\cap S_n| = 2$ and for all $1 \leq i < n$, $|S_i \triangle S_{i+1}| = 2$. 
\item For all odd $1 < i < n$, $||\Pi(S_i, S_{i-1}) - \Pi(S_i, S_{i+1})||_1 \leq 1/m^2$.
\item For all even $1 < i < n$, $||\Pi(S_{i-1},S_i) - \Pi(S_{i+1},S_i)||_1 \leq 1/m^2$.
\end{itemize}

\end{proposition}

\begin{proof}
Assume for contradiction that such a chain of sets exists. By Lemma~\ref{lem:protocol}, there exist functions $P(\cdot, \cdot)$ and $Q(\cdot, \cdot)$ such that $\Pr[\Pi(X, Y) = z] = P(z, X) \cdot Q(z, Y)$. We thus have (simply by expanding the bottom two hypotheses).
\begin{align*}
\sum_z \abs{P(z, S_i) \cdot (Q(z, S_{i-1}) - Q(z, S_{i+1}))} &\le \frac{1}{m^2} \text{ (} i \text{ odd)};\\
\sum_z \abs{Q(z, S_i) \cdot (P(z, S_{i-1}) - P(z, S_{i+1}))} &\le \frac{1}{m^2} \text{ (} i \text{ even)}.
\end{align*}

For notational simplicity for the remainder of the proof, we will denote by $a_i(z):= P(z, S_i)$ when $i$ is even, and $a_i(z) := Q(z, S_i)$ when $i$ is odd. Observe that the above equations are then simply:
\[\sum_z \abs{a_i(z)(a_{i - 1}(z) - a_{i + 1}(z))} \le \frac{1}{m^2},\quad 1 < i < n.\]
Finally, let $s(z) := \sum_{1 < i < n} \abs{a_i(z)(a_{i - 1}(z) - a_{i + 1}(z))}$. The following lemma bounds $a_1(z)(a_2(z) - a_n(z))$ in terms of $s(z)$. Recall that $a_1(z)\cdot a_2(z) = \Pr[\Pi(S_1,S_2)= z]$ and $a_1(z) \cdot a_n(z) := \Pr[\Pi(S_1,S_n) = z]$. So the lemma is bounding some term having to do with the difference between $\Pi(S_1,S_2)$ and $\Pi(S_1,S_n)$ in terms of the sums of differences between $\Pi(S_i, S_{i-1})$ and $\Pi(S_i, S_{i+1})$.

\begin{lemma}\label{lem:math}
For all $z$, we have $a_1(z)(a_2(z) - a_n(z)) \le ms(z)$.
\end{lemma}

\begin{proof}
To ease notational burden through this proof, we will drop the parameter $z$ (since the proof is independent of $z$). In particular, we will use terms $a_i$ throughout the proof and $s$, where $s:= \sum_{1 < i < n} |a_i\cdot (a_{i-1} - a_{i+1})|$. Observe that $a_i(a_{i-1}-a_{i+1}) \leq s$ for all $i$. 

If $a_1a_2 \leq ms$, then the lemma statement follows trivially, as $a_1, a_n \geq 0$. So now consider when that $a_1a_2 > ms$. In this case, we will define new $b_1,\ldots, b_n$ for which $b_i \leq a_i$ for all $i$, and analyze these instead. 

To this end, define $b_1 = a_1$, $b_2 = a_2$, and for $2 < i \leq n$ define $b_i:= b_{i-2} - s/b_{i-1}$. Observe that the $b_i$s satisfy the following equality (to see this, substitute $b_{i-1} - s/b_i$ for $b_{i+1}$):
\begin{equation}\label{eq:bs}
b_ib_{i+1} = b_i b_{i-1} - s.
\end{equation}

Since $b_1b_2 = a_1a_2 > ms$, this means that for all $i < n$ we have $b_i b_{i + 1} > (m - i + 1)s > 0$ (recall that $n = \frac{m}{2}$). Since $b_1$ and $b_2$ are strictly positive, we conclude that all $b_i$ are strictly positive as well. Now, we claim that $b_i \le a_i$ for all $i$. Indeed, this is true for $i = 1, 2$. We now prove this by induction for $i > 2$. Assume for inductive hypothesis that $b_j \leq a_j$. Then $b_{j + 1} = b_{j - 1} - \frac{s}{b_j} \le a_{j - 1} - \frac{s}{a_j} \le a_{j + 1}$. The last step follows from the equation $a_j(a_{j - 1} - a_{j + 1}) \le s$.

Now, we wish to prove further properties of the $b_i$s towards our end goal. We show the following inequality by induction:

\begin{equation}\label{eq:bs2}
b_1(b_2-b_{2i}) \le s \parens{\frac{m}{m - 1} + \frac{m}{m - 3} + \dots + \frac{m}{m - (2i - 3)}}.
\end{equation}
\begin{proof}[Proof of Equation~\ref{eq:bs2}]

Observe first for $i=2$ that we have:
\[b_1(b_2 - b_4) = b_1 \parens{b_2 - \parens{b_2 - \frac{s}{b_3}}} =\frac{b_1}{b_3}\cdot s = s \cdot \frac{b_1}{b_1 - \frac{s}{b_2}} = s \cdot \frac{b_1b_2}{b_1b_2 - s} = s \cdot \frac{a_1a_2}{a_1a_2 - s}.\]
Now, $\frac{a_1a_2}{a_1a_2 - s}$ is decreasing in $a_1a_2$, and we assumed already that $a_1a_2 > ms$, so we have
\[b_1(b_2-b_4) = s \cdot \frac{a_1a_2}{a_1a_2 - s} < s \cdot \frac{ms}{ms - s} = s \cdot \frac{m}{m - 1}.\]
This proves the base case ($i=2$). Now assume for inductive hypothesis that the inequality holds for $i - 1$. We then have:
\begin{align*}
b_1(b_2 - b_{2i}) &= b_1(b_2-b_{2(i-1)}) + b_1(b_{2(i-1)} - b_{2i}) = b_1(b_2-b_{2(i-1)}) + b_1 \cdot \frac{s}{b_{2i - 1}}.
\end{align*}
The last step here follows from Equation~\ref{eq:bs}. From here we continue with:
\begin{eqnarray*}
b_1(b_2-b_{2(i-1)}) + b_1 \cdot \frac{s}{b_{2i - 1}} &=& b_1(b_2-b_{2(i-1)}) + b_1 \cdot \frac{s}{b_{2i - 3} - \frac{s}{b_{2i - 2}}}\\ 
&=& b_1(b_2-b_{2(i-1)}) + b_1 \cdot \frac{b_{2i - 2}s}{b_{2i - 3}b_{2i - 2} - s}.
\end{eqnarray*}
From here, we now apply Equation~\eqref{eq:bs} to the term in the denominator $2i-4$ times. That is, we successively replace $b_j b_{j + 1}$ with $b_j b_{j - 1} - s$, $2i-4$ times. This leaves us with:
\begin{equation}
b_1(b_2-b_{2i}) = b_1(b_2-b_{2(i-1)}) + b_1 \cdot \frac{b_{2i - 2}s}{b_1b_2 - (2i - 3)s}.
\end{equation}

Now, it is easy to see from the definition of $b_i$s that for all $j$, $b_{j + 2} \le b_j$. As such, we also have $b_{2i - 2} \le b_2$. Finally, note that the denominator ($b_1b_2-(2i-3)s$) is positive (since we are in the case that $b_1b_2 =a_1a_2 > ms$), so we can now write:
\begin{align*}
b_1(b_2-b_{2i}) &\le b_1(b_2-b_{2(i-1)})+ s \cdot \frac{b_1b_2}{b_1b_2 - (2i - 3)s} \leq b_1(b_2-b_{2(i-1)}) + s \cdot \frac{ms}{ms - (2i - 3)s}\\
&\le s \parens{\frac{m}{m - 1} + \frac{m}{m - 3} + \dots + \frac{m}{m - (2i - 3)}}.
\end{align*}
The penultimate step above follows from the fact that $\frac{b_1b_2}{b_1b_2 - (2i - 3)s}$ decreasing in $b_1b_2$ and we have assumed that $b_1b_2 > ms$. The last step follows by the inductive hypothesis. This completes the proof of Equation~\eqref{eq:bs2}.
\end{proof}

Now we return to the proof of Lemma~\ref{lem:math} with Equation~\eqref{eq:bs2} in hand. Now, since $a_n \ge b_n$, we have:
\begin{align*}
a_1(a_2 - a_n) = b_1(b_2 - a_n) &\le b_1(b_2 - b_n) \le ms \parens{\frac{1}{m - 1} + \frac{1}{m - 3} + \dots + \frac{1}{m - \parens{\frac{m}{2} - 3}}}.
\end{align*}

In the right-most term, observe that there are $m/4-1$ terms in the sum, each of which are at most $2/m$. As such, the total sum of these terms is $\leq 1$. Therefore, the right-hand-side above is at most $ms$, and $a_1(a_2-a_n) \leq ms$, as desired. This concludes the proof of Lemma~\ref{lem:math}. 
\end{proof}

Now, we return to the proof of Proposition~\ref{prop:correctness}. Recall that the purpose of Lemma~\ref{lem:math} is to claim that the random variables $\Pi(S_1,S_2)$ and $\Pi(S_1,S_n)$ are not that different (which would contradict that $\pi$ is correct with probability at least $0.8$ on both $(S_1,S_2)$ and $(S_1,S_n)$). Observe by Lemma~\ref{lem:math} that:
\begin{eqnarray*}
&&\sum_{z, a_2(z) \ge a_n(z)} a_1(z)\left(a_2(z) - a_n(z)\right)
\le m \left(\sum_{z, a_2(z) \ge a_n(z)} s(z)\right) 
\le m \sum_z s(z)\\
&=& m \sum_{i = 2}^{n - 1} \sum_z \abs{a_i(z)(a_{i - 1}(z) - a_{i + 1}(z))} \le \frac{m(n - 2)}{m^2} \le \frac{1}{2}.
\end{eqnarray*}
Above, recall that $n:=m/2$, and that we have assumed for contradiction that $\sum_z |a_i(z) (a_{i-1}(z)-a_{i+1}(z)|\leq 1/m^2$ for all $i$. (Recall that $\sum_z |a_i(z) (a_{i-1}(z) - a_{i+1}(z)| := ||\Pi(S_i, S_{i-1}) - \Pi(S_i, S_{i+1})||_1$ when $i$ is odd, or $||\Pi(S_{i-1}, S_i) - \Pi(S_{i+1},S_i)||_1$ when $i$ is even, both of which are assumed to be $\leq 1/m^2$ in the proposition statement.)

Invoking that $a_1(z)a_2(z) = \Pr[\Pi(S_1,S_2) = z]$, and $a_1(z) a_n(z) = \Pr[\Pi(S_1,S_n) = z]$, we further conclude that:
\begin{eqnarray*}
\sum_{z, \Pr[\Pi(S_1,S_2) = z] \ge \Pr[\Pi(S_1,S_n) = z]} \abs{\Pr[\Pi(S_1,S_2) =z] - \Pr[\Pi(S_1,S_n)=z]}\\ 
=\sum_{z, a_2(z) \ge a_n(z)}a_1(z) (a_2(z) - a_n(z)) \le \frac{1}{2}.
\end{eqnarray*}

But note that $\sum_z \Pr[\Pi(S_1,S_2) = z] = \sum_z \Pr[\Pi(S_1,S_n) = z]= 1$, which means that:

\[\sum_{z, \Pr[\Pi(S_1,S_2)=z]\ge \Pr[\Pi(S_1,S_n) = z]} \abs{\Pr[\Pi(S_1,S_2)=z]-\Pr[\Pi(S_1,S_n)=z]}\]
\[= \sum_{z, \Pr[\Pi(S_1,S_2)=z]< \Pr[\Pi(S_1,S_n) = z]} \abs{\Pr[\Pi(S_1,S_2)=z]-\Pr[\Pi(S_1,S_n)=z]},\]

so in fact both sums are $\leq 1/2$, and we may now conclude that:

\begin{equation}\label{eq:almostthere!}
||\Pi(S_1,S_2) - \Pi(S_1,S_n)||_1 \leq 1.
\end{equation}

Now, we are ready to wrap up the proof of Proposition~\ref{prop:correctness}, as Equation~\ref{eq:almostthere!} asserts that the random variables $\Pi(S_1,S_2)$ and $\Pi(S_1,S_n)$ are too similar in order for protocol $\pi$ to be correct with probability at least $0.8$ on both $(S_1,S_2)$ and also $(S_1,S_n)$. 

Indeed, let $Z_b$ be the set of transcripts that output $b$. Then because $\pi$ is correct with probability at least $0.8$, $\Pi(S_1,S_n)$ assigns mass at least $0.8$ to $z \in Z_1$, and $\Pi(S_1,S_2)$ assigns mass at most $0.2$. So terms in $Z_1$ alone contribute at least $0.6$ to the difference. In addition, $\Pi(S_1,S_n)$ assigns mass at most $0.2$ to $z \in Z_0$, and $\Pi(S_1,S-2)$ assigns mass at least $0.8$. So terms in $Z_0$ contribute at least $0.6$ to the difference. Therefore, because $\pi$ is correct, we must have $||\Pi(S_1,S_2) - \Pi(S_1,S_n)||_1 \geq 1.2$, contradicting our conclusion above.

This concludes the proof of Proposition~\ref{prop:correctness}. To recap: we first showed how to relate the statistical difference between $\Pi(S_1,S_2)$ and $\Pi(S_1,S_n)$ to the statistical distance between adjacent pairs $\Pi(S_i, S_{i+1})$ and $\Pi(S_{i+2},S_{i+1})$ (Lemma~\ref{lem:math}). Lemma~\ref{lem:math} required a decent amount of math, the bulk of which is in proving Equation~\eqref{eq:bs2}. With Lemma~\ref{lem:math}, we then obtained a contradiction: the statistical difference between $\Pi(S_1,S_2)$ and $\Pi(S_1,S_n)$ cannot be small if $\pi$ is to possibly be correct on both $(S_1,S_2)$ and $(S_1,S_n)$, because the answers must be different with good probability.

\end{proof}

Proposition~\ref{prop:correctness} claims that no chains of the proposed form can exist for any private-randomness protocol which solves all instances of \fs\ with probability $0.8$. Now we will claim that any protocol which has low information complexity with respect to $\mu$ must have some chain of the proposed form. To this end, we first give a few definitions:
\begin{defin}
Let $n = \frac{m}{2}$ as before.
\begin{itemize}
\setlength{\itemsep}{0pt}
\item A \emph{link} is an ordered pair $(X, Y)$ of sets of size $n + 1$ such that $\abs{X \triangle Y} = 2$.
\item A \emph{chain} is a sequence of sets $T_1, \dots, T_n$ such that $(T_i, T_{i + 1})$ is a link for $1 \le i < n$, and $\abs{T_1 \cap T_n} = 2$.
\item Given a link $(X, Y)$, recall that $\Pi(X,Y)$ is the random variable denoting the transcript when $\pi$ is run on input $(X, Y)$. Further define $\Pi(X,?)$ to be the random variable which first samples a uniformly random $Y$ such that $(X,Y)$ is a link (note that this is simply sampling a uniformly random element of $X$ to remove, and a uniformly random element $\notin X$ to add), and then samples $\Pi(X,Y)$. Define $\Pi(?,Y)$ similarly.
\item A link $(X, Y)$ is \emph{broken} if $\mathbb{D}(\Pi(X,Y) \parallel \Pi(X,?)) > \frac{1}{8m^4}$ or if $\mathbb{D}(\Pi(X,Y) \parallel \Pi(?,Y)) > \frac{1}{8m^4}$. Here, $\mathbb{D}(\cdot \parallel \cdot)$ represents Kullback-Leibler divergence from Definition~\ref{def:kl}. 
\item A chain $T_1, \dots, T_n$ is \emph{broken} if for some odd $i$, $(T_i, T_{i + 1})$ or $(T_i, T_{i - 1})$ is a broken link. These $n - 1$ links are the chain's \emph{structural links}.
\end{itemize}
\end{defin}

\begin{proposition}\label{prop:brokenlinks}
Let $\pi$ be a protocol with $IC_\mu(\pi) \leq \frac{1}{8m^4 N}$. Then at most a $1/N$ fraction of links are broken. 
\end{proposition}
\begin{proof}
Let's first compute $I(\Pi(X,Y); X|Y)$ when $(X,Y)$ are drawn from $\mu$. In order to draw from $\mu$, we may first draw $X$ uniformly at random from all sets of size $n+1$, and then draw $Y$ uniformly at random among all sets of size $n+1$ with $|X \triangle Y| = 2$ (note that there are $(n+1)(n-1)$ such sets: pick an element in $X$ to kick out and an element $\notin X$ to add). So we get that:

\begin{align*}
I(\Pi(X,Y); X \mid Y) &= \frac{1}{\binom{m}{n + 1}} \sum_{\substack{T \subseteq M \\ \abs{T} = n + 1}} I(\Pi(X,Y); X \mid Y = T)\\
&= \frac{1}{\binom{m}{n + 1}} \sum_{\substack{T \subseteq M \\ \abs{T} = n + 1}} \sum_{S, \text{ } (S, T) \text{ is a link}} \frac{1}{(n + 1)(n - 1)} \mathbb{D}(\Pi(S,T) \parallel \Pi(?, T)).
\end{align*}

The first step above is just expanding the definition of conditional mutual information. The second step requires some further explanation. First, note that the number of $S$ such that $(S, T)$ is a link is $(n + 1)(n - 1)$ (one of the $n + 1$ elements of $T$ needs to be kicked out and one of the $n - 1$ elements of $T$ needs to be put in to form $S$), and each of these sets are equally likely to be drawn from $\mu$. The second step is again just unraveling the definition of mutual information via Fact~\ref{fact:KL}. 

Now, the above sum can be written to directly sum over all links. That is:

$$I(\Pi(X,Y); X \mid Y) = \frac{1}{\binom{m}{n+1}\cdot (n+1)(n-1)} \cdot \sum_{\text{links } (S,T)} \mathbb{D}(\Pi(S,T)||\Pi(?,T)).$$

Identical math concludes that $I(\Pi(X,Y);Y\mid X) =\frac{1}{\binom{m}{n+1}\cdot (n+1)(n-1)} \cdot \sum_{\text{links } (S,T)} \mathbb{D}(\Pi(S,T)||\Pi(S,?))$. This means that:

\begin{equation}
IC_\mu(\pi) = \frac{1}{\binom{m}{n+1} \cdot (n+1)(n-1)} D(\Pi(S,T) \parallel \Pi(?, T) )+ D(\Pi(S,T) \parallel \Pi(S,?)).
\end{equation}

Observe that this is simply the average over all links of $D(\Pi(S,T) \parallel \Pi(?, T) )+ D(\Pi(S,T) \parallel \Pi(S,?))$. So if a $1/N$ fraction of all links are broken, then the average of this quantity over all links is at least $\frac{1}{8m^4 N}$, meaning that $IC_\mu(\pi) \geq \frac{1}{8m^4 N}$. 

\end{proof}

Now, we claim that if most links are not broken, there must exist an entire chain which is not broken.
\begin{lemma}\label{lem:brokenlinks}
If fewer than $\frac{2}{m}$ links are broken, there exists an unbroken chain.
\end{lemma}

\begin{proof}
We first claim that the number of chains is $\frac{m!}{2}$. Indeed, consider the structure of a chain $T_1, \dots, T_n$. Observe that there are exactly two items that persist in $\cap_i T_i$. Moreover, each $T_i$ adds a unique element to $T_{i-1}$ (that previously wasn't added) and removes another one (which previously wasn't removed). So if we order the $m$ elements so that the first two elements are in $\cap_i T_i$, and the next $n-1$ elements are added in $T_2, T_3,\ldots, T_n$, and the next $n-1$ elements are removed in $T_2,\ldots, T_n$, this defines a chain. The chain defined is invariant under flipping the order of the first two elements, but modulo this, each ordering defines a unique chain. 

As counted previously, the total number of links is $\binom{m}{n + 1}(n - 1)(n + 1)$: for a link $(X, Y)$ there are $\binom{m}{n + 1}$ choices for $X$ and $(n - 1)(n + 1)$ choices for $Y$ conditioned on this. Each chain consists of $n - 1$ structural links, so the number of pairs $(\mathcal{C}, L)$ such that $L$ is a link in chain $\mathcal{C}$ is $\frac{m!}{2}(n - 1)$. By symmetry, every link is a structural link of the same number of chains, so for a given link, the number of chains which contain it as a structural link is:

\[\frac{m!(n - 1)}{2 \binom{m}{n + 1}(n - 1)(n + 1)} = \frac{m!}{2 \binom{m}{n + 1}(n + 1)}.\]

Thus, for all $\frac{m!}{2}$ chains to be broken, there must be at least $m!/2$ pairs $(\mathcal{C},L)$ such that $\mathcal{C}$ is a chain and $L$ is a broken structured link of $\mathcal{C}$. By the above counting, this means there must be at least $\binom{m}{n + 1}(n + 1)$ broken links. But this is $\frac{1}{n - 1}= \frac{1}{m/2-1}>2/m$ fraction of all of the links. So if fewer than this are broken, there must exist an unbroken chain.
\end{proof}

Now, we want to claim that an unbroken chain exactly satisfies the hypotheses of Proposition~\ref{prop:correctness}. This is the last step in wrapping up the proof of Theorem~\ref{biglemma}. 

\begin{lemma}\label{lem:unbroken} Let $S_1,\ldots, S_n$ be an unbroken chain. Then for all odd $1< i < n$, $||\Pi(S_i, S_{i-1}) - \Pi(S_i, S_{i+1})||_1 \leq 1/m^2$, and for all even $1< i < n$, $||\Pi(S_{i-1}, S_i) - \Pi(S_{i+1},S_i)||_1 \leq 1/m^2$. 
\end{lemma}

\begin{proof}
Because the chain is unbroken, we get that for all odd $i > 1$, $\mathbb{D}(\Pi(S_i, S_{i+1} \parallel \Pi(S_i,?))$ and $\mathbb{D}(\Pi(S_i,S_{i-1}) \parallel \Pi(S_i,?))$ are both at most $\frac{1}{8m^4}$. Now, Pinsker's inequality allows us to conclude that $||\Pi(S_i, S_{i+1}) - \Pi(S_i, ?)||_1 \leq \sqrt{\frac{2}{8m^4}} = \frac{1}{2m^2}$ for all $i$, and also that $||\Pi(S_i, S_{i+1})-\Pi(?,S_{i+1})||_1 \leq \frac{1}{2m^2}$ for all $i$. By the triangle inequality, we therefore conclude that for odd $i$:

\[||\Pi(S_i, S_{i+1})-\Pi(S_i, S_{i-1})||_1 \leq 1/m^2\]

And for even $i$:

\[||\Pi(S_{i+1}, S_i) - \Pi(S_{i-1},S_i)||_1 \leq 1/m^2.\]

\end{proof}
Now, we are finally ready to conclude the proof of Theorem~\ref{biglemma}. We have two contradictory lines of argument: on one hand, if $\pi$ uses private randomness and is correct on every input with probability at least $0.8$, then Proposition~\ref{prop:correctness} combined with Lemma~\ref{lem:unbroken} claims that no unbroken chain can exist. However, Proposition~\ref{prop:brokenlinks} combined with Lemma~\ref{lem:brokenlinks} claims that if $IC_\mu(\pi) < \frac{1}{4m^5}$, then there is an unbroken chain. So no private-randomness protocol can simultaneously have $IC_\mu(\pi) < \frac{1}{4m^5}$ and correctly solve \fs\ on every input with probability at least $0.8$. From Lemma~\ref{publicrandlemma} it also follows that no protocol, whether making use of public randomness or not, can simulatenously have $IC_\mu(\pi) < \frac{1}{4m^5}$ and correctly solve \fs\ on every input with probability at least $0.8$.
\end{proof}

\subsection{Wrapping everything up}
Now, we can complete the proof of Theorem~\ref{thm:randomized} by combining everything together. We first need to combine Theorem~\ref{biglemma} and Theorem~\ref{thm:braverman} to conclude that randomized protocols for \efs\ require exponential communication.

\begin{theorem}\label{thm:efs} For all $m,\ell,k$ such that $\ell \leq (1-c)\log_3(m)$ for some constant $c > 0$ and $k \in (m,\text{exp}(\frac{m^c}{2\log_3(m)})$, any randomized protocol $\pi$ that solves \efs\ with probability at least $2/3$ on all instances (which satisfy the promise) has $CC(\pi) = \Omega(k/m^5) $.
\end{theorem}
\begin{proof}
We already have a distribution $\mu$ for which any protocol $\pi$ solving \fs\ with probability at least $0.8$ on all inputs has $IC_\mu(\pi) > \frac{1}{4m^5}$. It is also clear that $\efs(\mathcal{X},\mathcal{Y}) = \bigvee_{i \in [k]} \fs(X_i, Y_i)$. So we just need to check the details with respect to the promises. In particular, we just need to see for what values of $(k,z)$, $\mu$ is and the \fs\ promise are $(k,z)$-safe with respect to the \efs\ promise.

Observe that conditions (1) and (2) of the \efs\ promise are trivially satisfied. So we just need to check conditions (3) and (4). Note that these conditions depend only on Alice's sets and, separately, Bob's sets, and not on how Alice's and Bob's sets interact. Now, the \efs\ promise is invariant under permutations of $M$, and also under permutations of the indices (i.e. if $(X_1,Y_1),\ldots, (X_k,Y_k)$ satisfy the \efs\ promise, then so do $(X_{\sigma(1)},Y_{\sigma(1)}),\ldots, (X_{\sigma(k)},Y_{\sigma(k)})$ for any permutation $\sigma$ from $[k]$ to $[k]$). Therefore, we may treat Alice as having i.i.d. subsets of $M$ of size $\frac{m}{2} + 1$, and likewise with Bob. We are interested in bounding the probability that Alice's sets do not satisfy conditions (3) and (4); then, a simple union bound will give us the probability that the \efs\ promise is satisfied.



\begin{lemma}Let $\mathcal{X} = \langle X_1,\ldots, X_k\rangle$ be drawn so that each $X_i$ is an i.i.d. uniformly random set of size $m/2+1$. For any $x$,  let $\ell:= \log_3(m) - \log_3(x)$, and let $k\leq e^{\frac{x}{2\ell}}$. Then with probability at least $1-e^{-x/2}-ke^{-\Omega(m)}$, $\mathcal{X}$ is $\ell$-sparse.

\end{lemma}
\begin{proof}
Consider the following roundabout way to draw $X_i$: first place each element of $M$ in $X'_i$ independently with probability $2/3$. Then, if $|X'_i| \geq m/2+1$, let $X_i$ be a random subset of $X'_i$ of size $m/2+1$. If any $|X'_i| < m/2+1$, abort the entire process and consider it a failure. Observe that $X_i \subseteq X'_i$.

For a fixed item $j$, and fixed set $L$ of indices with $|L| = \ell$, the probability that $j \in \cup_{i \in L} X'_i$ is $1-(1/3)^\ell$. Because these events are independent, the probability that $M\subseteq \cup_{i \in L} X'_i$ is exactly $(1-(1/3)^\ell)^m$. Taking a union bound over all $\binom{k}{\ell}$ subsets we get that the probability the collection is \emph{not} $\ell$-sparse, conditioned on not failing initially, is at most:

$$\binom{k}{\ell}\left(1-1/3^\ell\right)^m \leq k^\ell \text{exp}\left(\frac{-m}{3^\ell}\right) = \text{exp}\left(\ell \ln(k) - x\right) \leq \text{exp}(-x/2).$$

Finally, observe that the expected number of items in each $X'_i$ is $2m/3$. So the probability that a single $X'_i$ fails to have $m/2+1$ elements is $e^{-\Omega(m)}$ by the Chernoff bound. Taking a union bound over all $k$ $X'_i$s accounts for the additional $ke^{-\Omega(m)}$ term.
\end{proof}

\begin{lemma}Let $\mathcal{X} = \langle X_1,\ldots, X_k\rangle$ be drawn so that each $X_i$ is an i.i.d. uniformly random set of size $m/2+1$. Then with probability at least $1-e^{-k/3^\ell}-e^{-\Omega(m)}$, for all sets $T$ of $|T| = \ell$, there exists an $X_i \supseteq T$. 

\end{lemma}

\begin{proof}
Let's again consider the following roundabout way to draw $X_i$: first place each element of $M$ in $X'_i$ independently with probability $1/3$. Then if $|X'_i| \leq m/2+1$, let $X_i$ be a random superset of $X'_i$ of size $m/2+1$. If any $|X'_i| > m/2+1$, abort the entire process and consider it a failure. Observe that $X_i \supseteq X'_i$. 

For a fixed set $T$ of size $\ell$, The probability that $X'_i$ contains $T$ is just $1/3^\ell$. So the probability that no $X'_i$ contains $T$ is $(1-1/3^\ell)^k \leq \text{exp}(-k/3^\ell)$. Again, the probability of failure is at most $ke^{-\Omega(m)}$, resulting in the lemma statement.
\end{proof}

Now to wrap up, we observe that the probability that the \efs\ promise is \emph{not} satisfied is at most $2(ke^{-\Omega(m)} + e^{-k/3^\ell} + e^{-x/2})$, where $\ell := \log_3(m) - \log_3(x)$ and $k\leq e^{\frac{x}{2\ell}}$. Here, the factor of $2$ comes from a union bound from the events that Alice's sets fail condition (3) or (4) and that Bob's sets fail condiiton (3) or (4).

When $x = m^c$, for any $c < 1$, we get that $\ell := (1-c)\log_3(m)$, and $k \leq e^{\frac{m^c}{2(1 - c)\log_3(m)}}$. As such, $ke^{-\Omega(m)} = o(1)$, $e^{-x/2} =o(1)$, and furthermore $e^{-k/3^\ell} = o(1)$ as long as $k = \Omega(m)$. 

Theorem~\ref{thm:braverman} combined with Theorem~\ref{biglemma} then implies that any protocol $\pi$ for \efs\ which succeeds on all inputs (satisfying the promise) with probability at least $0.8/(1-o(1))$ has $CC(\pi) \geq \frac{k}{4m^5}$. To replace $0.8/(1-o(1))$ with $2/3$ in the Theorem statement, observe that if we had a protocol with success probability $2/3$, we could repeat it independently a constant number of times and take a majority to get a protocol with success probability $0.8/(1-o(1))$.
\end{proof}

\begin{proof}[Proof of Theorem~\ref{thm:randomized}]
The proof of Theorem~\ref{thm:randomized} now follows immediately from Theorem~\ref{thm:efs} and Corollary~\ref{cor:efswelfare}. In particular, by plugging in $\ell = \log_3(m)/3$, we may take $c = 2/3$, and $k = 4m^5 \cdot e^{\sqrt{m}}$ in Theorem~\ref{thm:efs} to conclude that any randomized protocol guaranteeing a $(1/2+\frac{6}{\log(m)})$-approximation for \wm\ with probability at least $2/3$ on all instances requires communication $\Omega(e^{\sqrt{m}})$.

\end{proof}

\section{From XOS to Subadditive: The MPH Hierarchy}
\label{sec:MPH}
In this section, we explore the space between fractionally subadditive and subadditive functions via the \emph{MPH hierarchy}.

\subsection{MPH Preliminaries}

Let's first review the definitions.

\begin{definition}[Positive-Hypergraph~\cite{AbrahamBDR12}] A valuation function $v(\cdot)$ is PH-$k$ if there exists a non-negative set function $w(\cdot)$ such that for all $S$, $v(S) = \sum_{T \subseteq S, |T| \leq k} w(T)$.
\end{definition}

Observe that PH-$1$ functions are exactly additive functions.

\begin{definition}[Maximum-over-Positive-Hypergraphs~\cite{FeigeFIILS15}] A valuation function $v(\cdot)$ is MPH-$k$ if there exists a collection $\mathcal{F}$ of PH-$k$ valuation functions such that for all $S$, $v(S) = \max_{f \in \mathcal{F}}\{f(S)\}$.
\end{definition}

Observe that MPH-$1$ functions are exactly fractionally subadditive (XOS) functions.

The next observation follows directly from the definitions.

\begin{observation}\label{obs:MPH} Let $v(\cdot)$ be MPH-$k$. Then for every item set $S$, there exists a non-negative set function $w_S(\cdot)$ such that (i) $v(S) = \sum_{T \subseteq S} w_S(T)$; (ii) $v(S') \geq \sum_{T \subseteq S'} w_S(T)$ for every $S'\ne S$; and (iii) $w_S(T) = 0$ for every set $T$ of size $|T| > k$.
\end{observation}
\begin{proof}
	By definition, $v(\cdot)$ is a maximum over PH-$k$ functions. Let $f_S(\cdot)$ be the PH-$k$ function that is the arg-maximizer at $S$, i.e., $v(S) = f_S(S)$; then clearly $v(S') \geq f_S(S')$ for every $S'$. Because $f_S(\cdot)$ is PH-$k$, by definition there exists a non-negative set function $w(\cdot)$ for which all conditions are satisfied.
\end{proof}

The following lemma lower-bounds the MPH level of a set function $v(\cdot)$ based on the ratio between the sum of marginal contributions of the items to the grand bundle and the grand bundle's value.

\begin{lemma}
	\label{lem:MPH-LB}
	Let $v(\cdot)$ be any set function. Then $v(\cdot)$ is not MPH-k for $k < \frac{\sum_j v(M) - v(M \setminus \{j\})}{v(M)}$.
\end{lemma}

\begin{proof}
	Let $k$ be the level of $v(\cdot)$ in the MPH hierarchy (i.e., $v(\cdot)$ is MPH-$k$). Then by Observation~\ref{obs:MPH} there exists a non-negative set function $w(\cdot)$ such that (i) $\sum_T w(T) = v(M)$; (ii) $v(M\setminus \{j\}) \ge \sum_{T \subseteq M\setminus \{j\}} w(T)$ for all $j$; and (iii) $w(T)=0$ for every $T$ with $|T|>k$.
	Summing both sides of (ii) over all items $j$ yields:
	$$
	\sum_{j \in M} v(M \setminus \{j\}) \geq
	\sum_j \sum_{T \subseteq M\setminus \{j\}} w(T)
	= \sum_T w(T) (m- |T|)
	\geq  \sum_T w(T) (m- k)
	= v(M)(m-k).
	$$
	Rearranging the inequality yields $k v(M) \geq \sum_j v(M) - v(M \setminus \{j\})$, as required.
\end{proof}

\cite{FeigeFIILS15} show that every monotone valuation function is MPH-$m$, and that there exist subadditive functions that are not in MPH-$k$ for any $k < m/2$.
In the Appendix we show that this is tight; i.e., that every monotone subadditive valuation function is MPH-$\lceil\frac{m}{2}\rceil$
(Proposition \ref{pro:subadditive-MPH-level}).

Clearly, for any $k>1$, there exist MPH-$k$ functions that are not complement-free (i.e., subadditive). Indeed, even MPH-$2$ functions exhibit complementarities. Since we are interested in exploring the space of functions between XOS and subadditive, our results in this section will be for ``subadditive MPH-$k$'' functions, which belong to both classes simultaneously. The extra subadditive assumption is necessary for our results (Proposition~\ref{prop:nosubadd}).

\subsection{Our Results for Subadditive MPH-$k$}

A preliminary question to address is: what is the MPH level of the subadditive functions constructed in Section \ref{sec:construction}? It turns out to be quite high:

\begin{proposition}
	\label{pro:MPH-level-of-construct}
	For all $\mathcal{S}, \ell$, such that $f_\mathcal{S}^\ell(\cdot)$ is well-defined, $f_\mathcal{S}^\ell(\cdot)$ is not MPH-$k$ for $k < m/\ell$.
\end{proposition}

\begin{proof}
	For every item $j$, $f_\mathcal{S}^\ell(M) - f_\mathcal{S}^\ell(M\setminus \{j\}) \geq 1$, and $f_\mathcal{S}^\ell(M) = \ell$. Applying Lemma \ref{lem:MPH-LB} completes the proof.
\end{proof}

Proposition \ref{pro:MPH-level-of-construct} raises the possibility to improve upon the $1/2$-approximation ratio for subadditive functions which reside in lower MPH levels. The main result of this section is a $(\frac{1}{2}+\Omega(\frac{1}{\log k}))$-approximation in $\poly(m)$ communication for \wm\ when Alice and Bob have valuations that are both subadditive and MPH-$k$. This approximation guarantee is tight for sufficiently large $k$ by direct application of our previous bounds.

Our main technical contribution in this section is an \emph{oblivious} rounding protocol for the configuration LP when two bidders both have subadditive MPH-$k$ valuations. The protocol's performance gradually degrades with the level $k$, starting from $0.625$ for $k=2$ (Section \ref{sub:protocol-MPH-2}) and behaving like $\frac{1}{2}+\Omega(\frac{1}{\log k})$ in general (Section \ref{sub:protocol-MPH-k}). An alternative and simpler protocol with guarantee $\frac{1}{2}+\Omega(\frac{1}{\log m})$ was developed independently for subadditive valuations by Dobzinski (which makes use of the fact that XOS valuations pointwise $1/H_m$-approximate subadditive functions.\footnote{This protocol, which is not yet in print, was brought to our attention in a personal correspondence with Shahar Dobzinski.} It is an interesting open question to determine whether the simpler version can be extended to subadditive MPH-$k$ valuations with the $\frac{1}{2}+\Omega(\frac{1}{\log k})$ guarantee (for instance, by proving that XOS valuations pointwise $1/H_k$-approximate subadditive MPH-$k$ functions. Note, however, that the ``subadditive'' is necessary for this claim to possibly be true). 

We remark that the obliviousness of the protocol is known to be without loss, by the results of~\cite{FeigeFT16}: since the class of MPH-$k$ valuations is closed under convex combinations, by \cite{FeigeFT16} there exists an oblivious rounding scheme that achieves an approximation guarantee matching the integrality gap of the LP. In Section~\ref{sub:IG} we show the matching integrality gap of $\frac{1}{2}+\Omega(\frac{1}{\log k})$ for sufficiently large~$k$.%
We also remark that our rounding-based technique necessarily fails for MPH-$k$ valuations that are not subadditive. 
Indeed, even for two MPH-$2$ valuations, the integrality gap may be as large as~$\frac{1}{2}$ (Proposition \ref{pro:IG-2-MPH-2}).

\subsection{Notation and Key Lemma}
\label{sub:key-lemma}

Throughout this section, we overload the notation $v(S)$ as follows. When $S$ is a random set drawn from distribution $D$, we use $v(S)$ to denote $\mathbb{E}_{S\sim D}[v(S)]$. Also, if $X$ denotes either Alice or Bob, we use $\overline{X}$ to denote the other player (i.e. if $X$ = Alice, then $\overline{X}$ = Bob, and vice versa).

The following key lemma extends a well-known result of~\cite{Feige06} for XOS valuations to the MPH-$k$ hierarchy. 

\begin{lemma} \label{mph_lemma}
Let $v(\cdot)$ be an MPH-$k$ function, and let $S \subseteq M$ be a subset of items. If $T$ is a random set such that every $U \subseteq S$ with $|U| \leq k$ is contained in $T$ with probability at least $p$, then $v(T) \ge p\cdot v(S)$.
\end{lemma}

\begin{proof}
	Let $w_S(\cdot)$ be as promised from Observation~\ref{obs:MPH}. Then $v(S) = \sum_{T \subseteq S} w_S(T)$, $v(U) \geq \sum_{T \in U} w_S(T)$ for all $U$, and $w_S(T) = 0$ for all $T$ with $|T| > k$. Then we can conclude that:
	\[v(T) \ge \sum_{T', |T| \leq k} \Pr[T' \subseteq T]\cdot w_S(T') \ge p \sum_{T'\subseteq S', |T'| \leq k} w_S(T') = p\cdot v(S).\]
	The first inequality follows from the fact that nonzero weights only belong to sets $T'$ of size at most~$k$. The second follows because every subset of $S$ of size $k$ appears in $T$ with probability at least~$p$.
\end{proof}

It is also not hard to extend Lemma~\ref{mph_lemma} to functions which are ``close'' to MPH-$k$.

\begin{definition}[Pointwise $\beta$-approximation~\cite{DevanurMSW15}] A valuation function $v(\cdot)$ is pointwise $\beta$-approximated by a valuation class $\mathcal{W}$ if for any set $S \subseteq M$, there exists a valuation $w \in \mathcal{W}$ such that: (a) $\beta w(S) \geq v(S)$ and (b) for all $T \subseteq M$, $v(T) \geq w(T)$.

Similarly, we say that a class $\mathcal{V}$ is pointwise $\beta$-approximated by a class $\mathcal{W}$ if all $v \in \mathcal{V}$ are pointwise $\beta$-approximated by $\mathcal{W}$.
\end{definition}

\begin{corollary} \label{mph_lemma_gen}
	Let $v(\cdot)$ be pointwise $\beta$-approximated by MPH-$k$ functions, and let $S\subseteq M$ be a subset of items. If $T$ is a random set such that every $U \subseteq S$ with $|U| \leq k$ is contained in $T$ with probability at least $p$, then $v(T) \ge \frac{p}{\beta}v(S)$.
 \end{corollary}
\begin{proof}
Let $w(\cdot)$ be the MPH-$k$ function which pointwise $\beta$-approximates $v(\cdot)$ at $S$. Then we know that: $v(T) \geq w(T) \geq p \cdot w(S) \geq \frac{p}{\beta} v(S)$. The outer two inequalities are by definition of pointwise approximation. The inner inequality is by Lemma~\ref{mph_lemma}.
\end{proof}


\subsection{Protocol for Subadditive MPH-$2$ Valuations}
\label{sub:protocol-MPH-2}

Here, we describe our protocol specifically for subadditive MPH-$2$, as it conveys the main ideas.
Our protocol will proceed as follows. First, we will solve the configuration LP relaxtion (defined shortly) which finds the optimal fractional allocation. Then, we provide an oblivious rounding which takes the fractional solution to a distribution over allocations. Assuming that both $A(\cdot)$ and $B(\cdot)$ are subadditive MPH-$k$, the rounding will maintain at least a $\frac{1}{2}+\Omega(1/\log(k))$ fraction of the welfare.

Let's now recall the configuration LP (defined below for any $n$, we only use it for $n = 2$):
\begin{itemize}
\setlength{\itemsep}{0pt}
\item Variables: $x_i(S)$, for all bidders $i$ and subsets $S \subseteq M$.
\item Constraint: $x_i(S) \geq 0$ for all $i, S$.
\item Constraint: For all $i$, $\sum_{S} x_i(S) = 1$ (dual variable $u_i$).
\item Constraint: For all $j$, $\sum_{S \ni j} \sum_i x_i(S) \leq 1$ (dual variable $p_j$).
\item Maximizing: $\sum_{i, S} v_i(S) \cdot x_i(S)$.
\end{itemize}

It is clear that the configuration LP is indeed a fractional relaxation, as any integral allocation is feasible. Despite having exponentially many variables, there are only $n+m$ non-trivial constraints, and so the dual has only $n+m$ variables. We quickly remind the reader of the dual:
\begin{itemize}
\setlength{\itemsep}{0pt}
\item Variables: $u_i$ for all bidders $i$, $p_j$ for all items $j$.
\item Constraint: $u_i \geq v_i(S) - \sum_{j \in S} p_j$ for all bidders $i$ and subsets $S \subseteq M$.
\item Constraint: $p_j \geq 0$.
\item Minimize: $\sum_j p_j + \sum_i u_i$.
\end{itemize}

We remind the reader that a separation oracle for the dual can be achieved via a \emph{demand oracle} for each $v_i(\cdot)$ (recall that this oracle takes as input a set of prices $\vec{p}$ and outputs the set maximizing $v_i(S) - \sum_{j \in S} p_j$). So the dual can indeed be solved in polynomial communication (and indeed only requires demand queries of the valuations). Once the dual is solved, an optimal primal can be recovered as well (for further details, see~\cite{BlumrosenN07}). Most state-of-the-art communication protocols (including all those referenced in Table~\ref{table:prior}) are derived by first solving this LP, and then rounding. Indeed, our protocol follows this blueprint as well. Moreover, our rounding protocol is \emph{oblivious}: while we of course need demand-query access to the valuations in order to find the optimal fractional solution, once we have this fractional solution our rounding scheme never looks at the valuations again. The same rounding argument is guaranteed to work for \emph{all} subadditive MPH-$2$ valuations.

{\bf Oblivious rounding scheme for $k=2$.}
Draw player $X$ uniformly at random from $\{\text{Alice, Bob}\}$ (also referred to as $\{1,2\}$).
\begin{enumerate}
\setlength{\itemsep}{0pt}
\item Let $\vec{x}$ denote the input fractional allocation.
\item Let $D_X$ denote the distribution that takes set $S$ with probability $x_X(S)$.
	\item With probability $\lambda=0.5$, draw $S_X$ from $D_X$; give set $S_X$ to player $X$; and give set $M \setminus S_X$ to player $\bar{X}$.
	\item Otherwise (with probability $1-\lambda=0.5$), draw set $S'_A$ from $D_A$ and set $S'_B$ from $D_B$; give set $S'_A \cap S'_B$ to player $X$; and give set $M\setminus(S'_A \cap S'_B)$ to player $\bar{X}$.
\end{enumerate}

\begin{proposition}
	\label{pro:analysis-MPH-2} Let $C$ denote the optimal value of the configuration LP. Then when Alice and Bob are both subadditive MPH-$2$, the expected welfare of the above oblivious rounding scheme is $\geq 0.625 \cdot C$.
\end{proposition}

\begin{proof}
Let bundles $S_A,S_B$ be independent random draws from distributions $D_A,D_B$, respectively.
The first case of the oblivious rounding scheme achieves expected welfare of (in expectation over randomly sampling $X$ from $U(\{1,2\})$).
\begin{eqnarray}
&&\frac{1}{2}(v_A(S_A)+v_B(M\setminus S_A))  + \frac{1}{2}(v_A(M\setminus S_B)+v_B(S_B))\nonumber\\
&\ge & \frac{1}{2}(v_A(S_A)+v_B(S_B)) + \frac{1}{2}(v_A(S_A\setminus S_B)+v_B(S_B\setminus S_A)),\label{eq:rounding-case1}
\end{eqnarray}
where the inequality holds pointwise for every instantiation of $S_A,S_B$ and follows from monotonicity of $v_A,v_B$.
Similarly, the second case of the oblivious rounding scheme achieves expected welfare of:

$$\frac{1}{2}(v_A(S_A\cap S_B)+v_B(S_A\cap S_B)) + \frac{1}{2}(v_A(M \setminus(S_A \cap S_B)) + v_B(M \setminus(S_A \cap S_B))).$$
To analyze the latter case we use the following claim:

\begin{claim}
	\label{cla:not-in-intersect}
	Let $v_X$ be MPH-2 and let bundles $S_A,S_B$ be independent random draws from distributions $D_A,D_B$, respectively. Then $v_X(M\setminus(S_A \cap S_B))\ge \frac{1}{2}v_X(M)$.
\end{claim}

\begin{proof}[Proof of Claim \ref{cla:not-in-intersect}]
	For every pair of items $i, j \in M$, the probability that player $X$ gets the pair when allocated the random bundle $M\setminus(S_A \cap S_B)$ is $\ge \frac{1}{2}$. To see this, observe that for player $X$ not to get the pair, either $i$ or $j$ or both must be in $S_A \cap S_B$. By definition of $S_A,S_B$, the probability $\pr{i \in S_A \cap S_B}$ is equal to $(\sum_{S \ni i} x_1(S))\cdot (\sum_{S \ni i} x_2(S))$. We further know that $\sum_{S \ni i} x_1(S) + x_2(S) \leq 1$ by the constraints in the configuration LP. This means that $(\sum_{S \ni i} x_1(S))\cdot (\sum_{S \ni i} x_2(S)) \leq 1/4$, and therefore $\pr{i \in S_A \cap S_B} \le 1/4$ as well. By applying the union bound we get that $i$ and/or $j$ are in $S_A \cap S_B$ with probability $\le 1/4$. Applying Lemma \ref{mph_lemma} for MPH-$2$ valuations, player $X$'s expected value for the random bundle $M\setminus(S_A \cap S_B)$ is thus at least $\frac{1}{2} v_X(M)$.
\end{proof}

By Claim \ref{cla:not-in-intersect} the second case achieves expected welfare of at least
\begin{eqnarray}
\frac{1}{2}(v_A(S_A\cap S_B)+v_B(S_A\cap S_B)) + \frac{1}{4}(v_A(\Omega) + v_B(\Omega)).\label{eq:rounding-case2}
\end{eqnarray}
Summing up the contributions from \eqref{eq:rounding-case1} and \eqref{eq:rounding-case2} weighted by their respective probabilities $\lambda$ and $1-\lambda$, the total expected welfare of the oblivious rounding scheme is at least:
	\begin{eqnarray}
	&&\frac{\lambda}{2}(v_A(S_A)+v_B(S_B)) +\nonumber\\
	&&\frac{\lambda}{2}(v_A(S_A\setminus S_B) + v_B(S_B\setminus S_A)) + \frac{1-\lambda}{2}(v_A(S_A\cap S_B) + v_B(S_A\cap S_B)) +\nonumber\\
	&&\frac{1-\lambda}{4} (v_A(M) + v_B(M))\nonumber\\
	&\ge&\frac{\lambda}{2}(v_A(S_A) + v_B(S_B)) + \frac{\lambda}{2}(v_A(S_A) + v_B(S_B)) + \frac{1-\lambda}{4}(v_A(M)+v_B(M))\label{eq:subadd}\\
	&\ge&\frac{5}{8}(f_A(S_A) + f_B(S_B)) = 0.625\cdot C, \label{eq:monotone}
	\end{eqnarray}
where \eqref{eq:subadd} holds since $\lambda$ was chosen such that $\frac{1-\lambda}{2}=\frac{\lambda}{2}$, and by subadditivity of $v_A,v_B$ (importantly, note that~\eqref{eq:subadd} does \emph{not} necessarily hold without subadditivity), and \eqref{eq:monotone} holds since $\lambda=\frac{1}{2}$ and by monotonicity of $v_A,v_B$ (both inequalities hold pointwise for every instantiation of $S_A,S_B$).
This completes the proof of Proposition \ref{pro:analysis-MPH-2}.
\end{proof}

To recap the proof of Proposition \ref{pro:analysis-MPH-2}, the fact that the valuations are MPH-$2$ means that if a player ``loses'' in the second rounding case and is allocated the ``leftovers'', this player still gets at least half of her total value (see Claim \ref{cla:not-in-intersect}). The fact that the valuations are subadditive means that allocating the bundle in contention $S_A\cap S_B$ with some probability $p$ to player $X$, and allocating the bundle not in contention $S_X\setminus (S_A\cap S_B)$ with the same probability to the same player, is as good in terms of welfare as allocating $S_X$ to player $X$ with probability $p$ (see \eqref{eq:subadd}).

\begin{corollary}
	\label{cor:close-to-MPH-2}
	If Alice and Bob are $\beta$-pointwise approximated by subadditive MPH-$2$ valuations,
	then the expected welfare of the oblivious rounding scheme is $\ge (\frac{1}{2}+\frac{1}{8\beta})\omega$.
\end{corollary}

Corollary \ref{cor:close-to-MPH-2} follows by replacing Lemma \ref{mph_lemma} by Corollary \ref{mph_lemma_gen} in the proof of Proposition \ref{pro:analysis-MPH-2}.

\subsection{Protocol for Subadditive MPH-$k$ Valuations}
\label{sub:protocol-MPH-k}

We now generalize the oblivious rounding scheme above to subadditive MPH-$k$. The main idea behind the approach is similar to our protocol for MPH-$2$: the ``level 0'' protocol that one should try first is to simply draw sets $S_A$ and $S_B$ independently from $D_A, D_B$. Prioritize awarding items in $S_A \cap S_B$ to a uniformly random player, and give the leftovers to the other. Of course, this protocol might fail to beat a $1/2$-approximation. But the only way it fails is if both Alice and Bob have high expected value for the intersection $S_A \cap S_B$. If Alice and Bob have high expected value for the intersection, then we can instead just recurse within $S_A \cap S_B$. The following protocol and subsequent proof makes this formal.

{\bf Oblivious rounding scheme for general $k\ge 2$.}
Draw player $X$ uniformly at random from $\{\text{Alice, Bob}\}$.
Set $r=\lceil\log\log k\rceil$.
\begin{enumerate}
	\setlength{\itemsep}{0pt}
	\item Let $\vec{x}$ denote the input fractional allocation.
	\item Let $D_X$ denote the distribution that takes set $S$ with probability $x_X(S)$.
	\item With probability $\lambda=\frac{1}{2}$, 
	draw $S_X$ from $D_X$; give set $S_X$ to player $X$; and give set $M \setminus S_X$ to player $\bar{X}$.
	\item For $0\le q <r$, with probability $\lambda_q=\frac{\lambda}{2^{q+1}}$, draw $2^q$ sets $S_A^1,\dots,S_A^{2^q}$ i.i.d. from $D_A$ and $2^q$ sets $S_B^1,\dots,S_B^{2^q}$ i.i.d. from $D_B$; give set $S^q=S_A^1\cap \dots \cap S_A^{2^q} \cap S_B^1\cap \dots\cap S_B^{2^q}$ to player $X$, and give set $M\setminus S^q$ to player $\bar{X}$.
	\item Otherwise (with probability $\lambda_r=1-\lambda-\sum_{q=0}^{r-1}\lambda_q=\frac{1}{2^{r+1}}$), draw $2^r$ sets $S_A^1,\dots,S_A^{2^r}$ i.i.d. from $D_A$ and $2^r$ sets $S_B^1,\dots,S_B^{2^r}$ i.i.d. from $D_B$; give $S^r=S_A^1\cap \dots \cap S_A^{2^r} \cap S_B^1\cap \dots\cap S_B^{2^r}$ to player $X$, and give $M\setminus S^r$ to player $\bar{X}$.
\end{enumerate}

\begin{theorem}
	\label{thm:analysis-MPH-k}
	Let $C$ denote the optimal value of the configuration LP. Then when Alice and Bob are both subadditive MPH-$k$, the expected welfare of the above oblivious rounding scheme is $\geq (\frac{1}{2} + \Omega (\frac{1}{\log k}))\cdot C$.
\end{theorem}

\begin{proof}
	The proof generalizes that of Proposition \ref{pro:analysis-MPH-2} and proceeds by analyzing the contribution from each case of the oblivious rounding scheme. The first case of the scheme achieves expected welfare of at least $\frac{1}{2}(v_A(S_A)+v_B(S_B)) + \frac{1}{2}(v_A(S_A\setminus S_B)+v_B(S_B\setminus S_A))$ (identical to \eqref{eq:rounding-case1}).
	For every $0\le q \le r$, the corresponding case of the scheme achieves expected welfare of at least
	\begin{equation}
	\frac{1}{2}(v_A(S^q)+v_B(S^q)) + \frac{1}{2}(v_A(S^q\setminus S^{q+1})+v_B(S^q\setminus S^{q+1})).\label{eq:rounding-case-internal}
	\end{equation}
	To analyze the last case we use the following claim:
	
	\begin{claim}
		\label{cla:not-in-intersect-k}
		Let $v_X$ be an MPH-$k$ valuation and let bundles $S_A^1,\dots,S_A^{2^r}$ and $S_B^1,\dots,S_B^{2^r}$ be independent random draws from distributions $D_A$ and $D_B$, respectively.
		Let $S^r=S_A^1\cap \dots \cap S_A^{2^r} \cap S_B^1\cap \dots\cap S_B^{2^r}$.
		Then $v_X(M\setminus S^r)\ge (1-\frac{k}{4^{2^r}})v_X(M)$.
	\end{claim}
	
	\begin{proof}[Proof of Claim \ref{cla:not-in-intersect-k}]
		For every bundle of $k$ items, the probability that player $X$ gets this bundle when allocated the random bundle $M\setminus S^r$ is $\ge 1 - \frac{k}{4^{2^r}}$:
		For player $X$ not to get the bundle, at least one of its items must be in $S^r$. By definition of $S^r$, the probability $\pr{i \in S^r}$ is equal to $(\sum_{S \ni i} x_1(S))^{2^r}\cdot (\sum_{S \ni i} x_2(S))^{2^r}$, where $\sum_{S \ni i} x_1(S) + x_2(S) \le 1$ because the solution is feasible for the configuration LP. This means that $\pr{i \in S^r} \le \frac{1}{4^{2^r}}$, and by applying the union bound we get that at least one of the items in the bundle are in $S^r$ with probability $\le \frac{k}{4^{2^r}}$. Applying Lemma \ref{mph_lemma} for MPH-$k$ valuations, player $X$'s expected value for the random bundle $M\setminus S^r$ is thus at least $(1 - \frac{k}{4^{2^r}}) v_X(M)$.
	\end{proof}
	
	By Claim \ref{cla:not-in-intersect-k} the last case achieves expected welfare of at least
	\begin{eqnarray}
	\frac{1}{2}(v_A(S^r)+v_B(S^r)) + \frac{1}{2}(1-\frac{k}{4^{2^r}})(v_A(M) + v_B(M)).\label{eq:rounding-case-last}
	\end{eqnarray}
	
 	We can now sum up the contributions from the first case, the intermediate case \eqref{eq:rounding-case-internal}, and the last case \eqref{eq:rounding-case-last} weighted by their respective probabilities of $\lambda$, $\lambda_q$ for every intermediate case $0\le q<r$, and $\lambda_r$.
 	Notice that $\lambda_{r-1}=\lambda_{r}$.
	The weighted sum of the last \emph{two} cases is thus at least
	\begin{eqnarray}
	&&\frac{\lambda_{r-1}}{2}(v_A(S^{r-1})+v_B(S^{r-1})) + \frac{\lambda_{r-1}}{2}(v_A(S^{r-1}\setminus S^{r})+v_B(S^{r-1}\setminus S^{r})) \nonumber\\&+&
	\frac{\lambda_{r}}{2}(v_A(S^r)+v_B(S^r))
	+ \frac{\lambda_r}{2}(1-\frac{k}{4^{2^r}})(v_A(M) + v_B(M))\nonumber\\
	&\ge& \frac{\lambda_{r-2}}{2}(v_A(S^{r-1})+v_B(S^{r-1}))
	+ \frac{1}{2^{r+2}}(1-\frac{k}{4^{2^r}})\cdot C\nonumber\\
	&=& \frac{\lambda_{r-2}}{2}(v_A(S^{r-1})+v_B(S^{r-1}))
	+ \Omega(\frac{1}{\log k})\cdot C
	,\nonumber
	\end{eqnarray}
	Above, the first inequality follows by first observing that $\lambda_r = \lambda_{r-1}$, $\lambda_{r-2} = 2\lambda_{r-1}$, and using subadditivity of $v_A(\cdot), v_B(\cdot)$ to combine the first three terms together. The last equality follows by observing that $k/4^k < 1/2$, and that $r = \lceil \log \log k \rceil$.

	Now, we wish to continue by induction and sum the last $\tau$ cases for $\tau \leq r+1$. We claim that the contribution from the last $\tau$ cases is at least:

$$\frac{\lambda_{r-\tau}}{2} \cdot (v_A(S^{r-\tau+1}) + v_B(S^{r-\tau+1})) + \Omega(1/\log k) \cdot C.$$

We have already proven the base case: this holds when plugging in $\tau = 2$. Assume for inductive hypothesis that the claim holds for $\tau$, and we now prove it for $\tau+1$. Then the sum of the last $\tau+1$ terms is exactly:
\begin{eqnarray*}
\frac{\lambda_{r-\tau}}{2}\left(v_A(S^{r-\tau})+v_B(S^{r-\tau}) + v_A(S^{r-\tau}\setminus S^{r-\tau+1}) + v_B(S^{r-\tau}\setminus S^{r-\tau+1})\right) +&&\\
\frac{\lambda_{r-\tau}}{2} \cdot (v_A(S^{r-\tau+1}) + v_B(S^{r-\tau+1})) + \Omega(1/\log k) \cdot C &\geq&\\
\frac{\lambda_{r-\tau-1}}{2} \left(v_A(S^{r-\tau}) + v_B(S^{r-\tau}) \right) + \Omega(1/\log k) \cdot C.&&
\end{eqnarray*}

The inequality follows by definition of $\lambda_{r-\tau-1} := 2\lambda_{r-\tau}$, and because both $v_A(\cdot)$ and $v_B(\cdot)$ are subadditive. This proves the inductive step. This means that the last $r+1$ cases together contribute exactly (below, $\lambda_{-1}:=2\lambda_0= \lambda$:
$$
\frac{\lambda_{-1}}{2} \cdot (v_A(S^{0}) + v_B(S^{0})) + \Omega(1/\log k) \cdot C.
$$

Now, together with the weighted contribution of the first case, which is at least
$$
\frac{\lambda}{2}\left(v_A(S_A) + v_B(S_B) + v_A(S_A\setminus S_B) + v_B(S_B \setminus S_A)\right),
$$
we get a total expected welfare of at least (recall that $S^0 = S_A \cap S_B$):
$$ \frac{\lambda}{2} \cdot \left(v_A(S_A) + v_A(S_B) + v_A (S_A \setminus S_B) + v_B(S_B \setminus S_A) + v_A(S_A \cap S_B) + v_B(S_A \cap S_B)\right) + \Omega(1/\log k)\cdot C$$
$$ \geq \lambda \left(v_A(S_A) + v_B(S_B)\right) + \Omega(1/\log k) \cdot C = (1/2 + \Omega(1/\log k)) \cdot C.$$
This completes the proof of Theorem \ref{thm:analysis-MPH-k}.
\end{proof}

\begin{corollary}
	\label{cor:close-to-MPH-k}
	If Alice and Bob have valuations that are pointwise $\beta$-approximated by subadditive MPH-$k$ valuations, then the expected welfare of the oblivious rounding scheme is $\ge (\frac{1}{2}+\Omega(\frac{1}{\beta\log k}))\cdot C$.
\end{corollary}

\subsection{Integrality Gaps and Hardness}
\label{sub:IG}

In this section, we briefly derive integrality gaps and communication lower bounds for subadditive MPH-$k$ valuations based on previous constructions. We also show an integrality gap of $\frac{1}{2}$ for MPH-$2$ valuations that are not subadditive.

\begin{proposition}\label{prop:nosubadd}
	Let $k\ge 2$. The integrality gap of the configuration LP with two subadditive MPH-$k$ bidders is $\frac{1}{2} + \Theta(\frac{1}{\log k})$.
\end{proposition}

\begin{proof}[Proof Sketch]
The rounding algorithm presented above witnesses that the integrality gap is $1/2 + O(1/\log k)$, and the construction from Section~\ref{sec:construction} witnesses that the gap is $1/2+\Omega(1/\log k)$.

Specifically, let $S_1,\ldots, S_t$ be random sets of size $k/2$ that are all subsets of the same $K \subseteq M$ of $|K| = k$. Then with $t := \sqrt{k}$, $\ell := \log_2(k)/2$, $\{S_1,\ldots, S_t\}$ is $\ell$-sparse with probability $1-1/\poly(k)$. So consider the instance with $f_\mathcal{S}^\ell(\cdot) = v_1(\cdot) = v_2(\cdot)$. Then in this case, we know that the best achievable welfare is $\ell$ for an integral solution. However, the fractional solution which sets $x_i(\overline{S_j}) = 1/t$ for all $t$ is feasible for the configuration LP, and achieves welfare $2(\ell-1)$.
\end{proof}

\begin{proposition}
	\label{pro:IG-2-MPH-2}
	The integrality gap of the configuration LP with two (non-subadditive) MPH-$2$ bidders is $\frac{1}{2}$.
\end{proposition}

\begin{proof}
	Consider 4 items $a,b,c,d$. Alice has value $1$ for bundle $\{a,b\}$ and for bundle $\{c,d\}$ (as well as for any containing bundle); and Bob has value $1$ for bundle $\{a,c\}$ and for bundle $\{b,d\}$ (as well as for any containing bundle); all other values are $0$. These valuations are MPH-$2$ since they can be described as the maximum over $2$ hypergraphs, each with a positive hyperedge of size $2$ corresponding to one of the two desired bundles. The best fractional solution to the configuration LP is $x_1(\{a,b\})=x_1(\{c,d\})=x_2(\{a,c\})=x_2(\{b,d\})=\nicefrac{1}{2}$; one can check that all constraints are satisfied and the objective is $2$. The best integral solution however cannot achieve welfare better than $1$, completing the proof.
\end{proof}

\begin{proposition}
Let $k \geq 2$. There exists an absolute constant $C$ such that the (randomized) communication required to achieve a $(1/2 +C/\log k)$-approximation for \wm\ even when both Alice and Bob are subadditive MPH-$k$ is $\Omega(e^{\sqrt{k}})$.
\end{proposition}
\begin{proof}
This is a direct application of Theorem~\ref{thm:randomized}, after observing that any subadditive valuation defined on items $K \subseteq M$ (with zero value for all other items) is MPH-$|K|$.
\end{proof}

Observe that the above proposition further means that our approximation guarantees for subadditive MPH-$k$ are asymptotically tight among protocols with $\poly(m)$ communication whenever $k =\Omega(\log^3 m)$ (tighter calculations could get this down to $\log^{(1+\varepsilon)}m$ for any constant $\varepsilon > 0$, if desired). It remains open whether there is an impossibility result $< 3/4$ for two subadditive MPH-$2$ valuations ($3/4$ is implied already by the impossibility for MPH-1 = XOS~\cite{DobzinskiNS10}).

\section*{Acknowledgments}
The authors are grateful to Mark Braverman, Gillat Kol, Jieming Mao, Raghuvansh Saxena, and Omri Weinstein for numerous pointers towards lecture notes and prior work on information complexity, and helpful comments on this portion of the paper. The authors also thank Noga Alon for referring us to the notion of $\ell$-independent sets, and Shahar Dobzinski and Uri Feige for providing very helpful feedback and insight on welfare maximization.

The work of M. Feldman was partially supported by the European Research Council under the European Union's Seventh Framework Programme (FP7/2007-2013) / ERC grant agreement number 337122.

E. Neyman and S.M. Weinberg are supported by NSF CCF-1717899.

\appendix
\section{Missing Proofs}
\label{appx:missing-proofs}

\begin{proof}[Proof of Lemma \ref{lem:well-defined}]
	It is clear that $f^\ell_\mathcal{S}(X)$ is always defined at least once. The only way in which $f^\ell_\mathcal{S}(X)$ could be defined multiple times is if $\sigma_\mathcal{S}(X) < \frac{\ell}{2}$ (in which case $f^\ell_\mathcal{S}(X) = \sigma_\mathcal{S}(X)$) and $\sigma_\mathcal{S}(\overline{X}) < \frac{\ell}{2}$ (in which case $f^\ell_\mathcal{S}(X) = \ell - \sigma_\mathcal{S}(\overline{X})$).
	So assume for contradiction that both events hold, and let $X \subseteq \cup_{i=1}^{\ell/2-1} T_i$, and $\overline{X} \subseteq \cup_{i=1}^{\ell/2-1} U_i$, where each $T_i, U_i \in \mathcal{S}$. Note that $\sigma_\mathcal{S}(X), \sigma_\mathcal{S}(\overline{X}) < \ell/2$ implies that such $T_i, U_i$ must exist. But now consider that we can write $M = X \cup \overline{X}$ as a union of $\leq \ell -2$ elements of $\mathcal{S}$, contradicting that $\mathcal{S}$ is $\ell$-sparse.
\end{proof}

\begin{proof}[Proof of Lemma \ref{fsigmapropslem}]
	We show \ref{sigsubadd} similarly to~\cite{BhawalkarR11}. First, $\sigma_\mathcal{S}(\cdot)$ is monotone because if $X$ is a subset of $\cup_{i \in Y} S_i$, then $X' \subseteq X$ is also a subset of $\cup_{i \in Y} S_i$ (and therefore, any subcollection of $\mathcal{S}$ that covers $X$ also covers $X'$, and $\sigma_\mathcal{S}(X') \leq \sigma_\mathcal{S}(X)$). Otherwise if $X$ is not covered by $\mathcal S$, then $\sigma_\mathcal{S}(X)=\max\{\ell,k\}$ and certainly $\sigma_\mathcal{S}(X') \leq \sigma_\mathcal{S}(X)$. Second, $\sigma_\mathcal{S}$ is subadditive:
	Note that at least one of $X,W$ is not covered by $\mathcal{S}$ if and only if $X\cup W$ is not covered by $\mathcal{S}$, and in this case $\sigma_\mathcal{S}(X \cup W) = \max\{\ell,k\} \leq \sigma_\mathcal{S}(X) + \sigma_\mathcal{S}(W)$.
	Otherwise, consider the index sets $Y, Z$ witnessing $\sigma_\mathcal{S}(X)$ and $\sigma_\mathcal{S}(W)$, respectively (that is, $X \subseteq \cup_{i \in Y} S_i$ and $\sigma_\mathcal{S}(X) = |Y|$, ditto for $W$ and $Z$). Then $X \cup W \subseteq \cup_{i \in Y \cup Z} S_i$, and $\sigma_\mathcal{S}(X \cup W) \leq |Y| + |Z| = \sigma_\mathcal{S}(X) + \sigma_\mathcal{S}(W)$.
	
	\ref{lminus} follows directly from the definition of $f^\ell_\mathcal{S}$, and the fact that it is well-defined:~\ref{lminus} clearly holds for any set $X$ for which $f^\ell_\mathcal{S}(X)$ is defined in \ref{fpart1}, and also for any set $X$ for which $f^\ell_\mathcal{S}(X)$ is defined in \ref{fpart2} (as both $f^\ell_\mathcal{S}(X)$ and $f^\ell_\mathcal{S}(\overline{X})$ are $\ell/2$).
	
	\ref{feqsigma} holds because if $\sigma_\mathcal{S}(X) < \frac{\ell}{2}$ then $f^\ell_\mathcal{S}(X) = \sigma_\mathcal{S}(X)$ by definition. And if $f^\ell_\mathcal{S}(X) < \frac{\ell}{2}$, the only way this is possible (given the definition of $f^\ell_\mathcal{S}(\cdot)$) is if $f^\ell_\mathcal{S}(X)$ is defined in~\ref{fpart1} to be $\sigma^\ell_\mathcal{S}(X)$.
	
	\ref{gtl2} holds again because if $f^\ell_\mathcal{S}(X) > \frac{\ell}{2}$, the only way this is possible (given the definition of $f^\ell_\mathcal{S}(\cdot)$) is if $f^\ell_\mathcal{S}(X)$ is defined to be $\ell - \sigma_\mathcal{S}(\overline{X})$.
	
	\ref{coverUB} holds because by definition of $f^\ell_\mathcal{S}(\cdot)$, one of three possible cases must occur: either $f^\ell_\mathcal{S}(X) = \sigma_\mathcal{S}(X)$, or $f^\ell_\mathcal{S}(X) = \ell - \sigma_\mathcal{S}(\overline{X}) \le \sigma_\mathcal{S}(M) - \sigma_\mathcal{S}(\overline{X}) \le \sigma_\mathcal{S}(X)$ (using the $\ell$-sparsity of $\mathcal{S}$ by which $\sigma_\mathcal{S}(M)\ge \ell$ (which holds even if $M$ is not covered by $\mathcal{S}$) and subadditivity of $\sigma_\mathcal{S}$), or $f^\ell_\mathcal{S}(X) = \ell/2$ and $\sigma_\mathcal{S}(X)\ge \ell/2$.
\end{proof}

\begin{proof}[Proof of Corollary \ref{cor:construction-subadditive}]
	(Monotonicity) Let $X \subseteq T$ and suppose for contradiction that $f^\ell_\mathcal{S}(T) < f^\ell_\mathcal{S}(X)$. First suppose that $f^\ell_\mathcal{S}(T) < \frac{\ell}{2}$. By part \ref{feqsigma} of Lemma~\ref{fsigmapropslem}, we can conclude that $\sigma_\mathcal{S}(T) < \frac{\ell}{2}$, and therefore by part \ref{sigsubadd} (specifically, monotonicity of $\sigma_\mathcal{S}(\cdot)$), $\sigma_\mathcal{S}(X) < \frac{\ell}{2}$. Thus, by another application of \ref{feqsigma}, we get: $f^\ell_\mathcal{S}(X) = \sigma_\mathcal{S}(X) \le \sigma_\mathcal{S}(T) = f^\ell_\mathcal{S}(T)$, a contradiction.
	
	Next suppose that $f^\ell_\mathcal{S}(T) \ge \frac{\ell}{2}$. By assumption, this means that $f^\ell_\mathcal{S}(X) > \frac{\ell}{2}$, so by part \ref{gtl2} we have $\sigma_\mathcal{S}(\overline{X}) = \ell - f^\ell_\mathcal{S}(X)$. Since $\overline{T} \subseteq \overline{X}$, by \ref{sigsubadd} we have $\sigma_\mathcal{S}(\overline{T}) \le \sigma_\mathcal{S}(\overline{X}) = \ell  - f^\ell_\mathcal{S}(X) < \frac{\ell}{2}$. So by applying parts \ref{lminus} and \ref{feqsigma} we have $\ell - f^\ell_\mathcal{S}(T) = f^\ell_\mathcal{S}(\overline{T}) = \sigma_\mathcal{S}(\overline{T}) \le \ell - f^\ell_\mathcal{S}(X)$, implying that $f^\ell_\mathcal{S}(X) \le f^\ell_\mathcal{S}(T)$, a contradiction. Therefore, we have a contradiction in both cases, and we must have $f^\ell_\mathcal{S}(X) \le f^\ell_\mathcal{S}(T)$.
	
	(Subadditivity) Suppose for contradiction that $f^\ell_\mathcal{S}(X\cup T) > f^\ell_\mathcal{S}(X) + f^\ell_\mathcal{S}(T)$. We first show that $f^\ell_\mathcal{S}(X \cup T) > \frac{\ell}{2}$. Indeed, suppose that $f^\ell_\mathcal{S}(X \cup T) \le \frac{\ell}{2}$. Then $f^\ell_\mathcal{S}(X)$ and $f^\ell_\mathcal{S}(T)$ are both $<\frac{\ell}{2}$, so $f^\ell_\mathcal{S}(X) = \sigma_\mathcal{S}(X)$ and $f^\ell_\mathcal{S}(T) = \sigma_\mathcal{S}(T)$ by \ref{feqsigma}. So by subadditivity of $\sigma_\mathcal{S}(\cdot)$, we have $\sigma_\mathcal{S}(X \cup T) \le f^\ell_\mathcal{S}(X) + f^\ell_\mathcal{S}(T)$. Note finally that $f^\ell_\mathcal{S}(X \cup T) \le \sigma_\mathcal{S}(X \cup T)$ by part~\ref{coverUB}. Thus, $f^\ell_\mathcal{S}(X \cup T) \le f^\ell_\mathcal{S}(X) + f^\ell_\mathcal{S}(T)$, a contradiction.
	
	Now assume for contradiction that $f^\ell_\mathcal{S}(X \cup T) > \frac{\ell}{2}$.
	This means that $f^\ell_\mathcal{S}(X \cup T) = \ell - \sigma_\mathcal{S}(\overline{X \cup T})$ by~\ref{gtl2}. Observe also that $f^\ell_\mathcal{S}(X \cup T) \le \ell$. Since by assumption, $f^\ell_\mathcal{S}(X \cup T) > f^\ell_\mathcal{S}(X) + f^\ell_\mathcal{S}(T)$, at least one of $f^\ell_\mathcal{S}(X)$ and $f^\ell_\mathcal{S}(T)$ is $<\frac{\ell}{2}$; without loss of generality, assume that $f^\ell_\mathcal{S}(X) < \frac{\ell}{2}$, so $f^\ell_\mathcal{S}(X) = \sigma_\mathcal{S}(X)$ by \ref{feqsigma}.
	Using what we've concluded so far, we may rewrite $f^\ell_\mathcal{S}(X \cup T) > f^\ell_\mathcal{S}(X) + f^\ell_\mathcal{S}(T)$ as $\ell - \sigma_\mathcal{S}(\overline{X \cup T}) > \sigma_\mathcal{S}(X) + f^\ell_\mathcal{S}(T)$, i.e.,
	\begin{equation}
	\sigma_\mathcal{S}(\overline{X \cup T}) + \sigma_\mathcal{S}(X) + f^\ell_\mathcal{S}(T) < \ell.\label{eq:for-contradiction}
	\end{equation}
	We have that $\overline{T}\subseteq \overline{X \cup T} \cup X$ (De Morgan), and so
	$\sigma_\mathcal{S}(\overline{X \cup T}) + \sigma_\mathcal{S}(X) \ge \sigma_\mathcal{S}(\overline{T})$ by subadditivity of $\sigma_\mathcal{S}(\cdot)$. Plugging this observation into Equation \eqref{eq:for-contradiction}, we get
	$\sigma_\mathcal{S}(\overline{T}) + f^\ell_\mathcal{S}(T) < \ell$.
	But by parts \ref{sigsubadd} and \ref{coverUB} of Lemma \ref{fsigmapropslem}, $f^\ell_\mathcal{S}(X) = \ell - f^\ell_\mathcal{S}(\overline{X}) \ge \ell - \sigma_\mathcal{S}(\overline{X})$ for all $X$, a contradiction. We therefore conclude that $f^\ell_\mathcal{S}(\cdot)$ must be subadditive, as we have derived contradictions whether $f_{\mathcal{S}}^\ell(X \cup T) \leq \frac{\ell}{2}$ (first paragraph) or $f_{\mathcal{S}}^\ell(X \cup T) > \frac{\ell}{2}$ (just now).
\end{proof}

\begin{proof}[Proof of Lemma \ref{lem:l-indep-exist}]
	Consider the following randomized construction of $\mathcal{S}$: For each $i \in [k], j \in M$ independently, flip a fair coin. If heads, put $j\in S_i$. Otherwise, don't. We wish to show that the probability that $\mathcal{S}$ is $\ell$-independent is non-zero. We'll again use $S_i^1 := S_i$, and $S_i^0 := \overline{S_i}$.
	
	So fix any set $Y$ of $\ell$ indices, and any vector $y \in \{0,1\}^\ell$ (for simplicity of notation, index these $\ell$ bits using the indices of $Y$). We wish to consider the probability that $\cup_{i \in Y} S_i^{y_i} = M$. If there exists a single $Y, y$ such that $\cup_{i \in Y} S_i^{y_i} = M$, then $\mathcal{S}$ is \emph{not} $\ell$-independent. But if for all $Y, y$ $\cup_{i \in Y} S_i^{y_i} \neq M$, then $\mathcal{S}$ \emph{is} $\ell$-independent. So we simply wish to analyze the probability that this occurs for a fixed $Y, y$ and take a union bound.
	
	For a fixed $Y, y$, observe that each element $j \in M$ is in each $S_i^{y_i}$ independently with probability $1/2$. So the probability that $j$ is in \emph{some} $S_i^{y_i}$ is just $1-1/2^\ell$. Moreover, these events are independent across items $j$, so $\cup_{i \in Y} S_i^{y_i} = M$ with probability $(1-1/2^\ell)^m$. Now we wish to take a union bound over all $2^\ell$ choices of $y$ times $\binom{k}{\ell}$ choices of $Y$ to get that the probability that $\mathcal{S}$ is \emph{not} $\ell$-independent is upper bounded by (the final two steps use our choice of $\ell = \log_2(m) - \log_2(x)$ and $k = e^{x/\ell}$):
	
	\[2^\ell \binom{k}{\ell} \parens{1 - \frac{1}{2^\ell}}^m < \frac{2^\ell \cdot k^\ell}{\ell !} \parens{1 - \frac{1}{2^\ell}}^m < k^\ell \exp \parens{\frac{-m}{2^\ell}}=\exp\parens{\ell \cdot \ln(k) - x} = 1.\]
	
	As the probability that $\mathcal{S}$ is not $\ell$-independent is $<1$, we are guaranteed the existence of some $\mathcal{S}$ that is $\ell$-independent by the probabilistic method.
\end{proof}

\begin{proposition}
\label{pro:subadditive-MPH-level}
Every monotone and subadditive function over a set of $m$ items is MPH-$\lceil \frac{m}{2} \rceil $.
\end{proposition}

\begin{proof}
Let $f$ be a monotone subadditive function.
For every set $S \subseteq [m]$ we construct a positive hypergraph $G_S$ with weights $w_S$, as follows:
Let $S'$ be an arbitrary subset of $S$ of size $\min (|S|,\lceil \frac{m}{2} \rceil )$.
Set $w_S(S')=f(S')$ and $w_S(S \setminus S')= f(S)-f(S')$. All other hyperedges have weight 0.
Observe that in our construction, every hyperedge with a non-zero weight has size at most $\lceil \frac{m}{2} \rceil$.
For every set $T$, let $f_S(T)$ denote the value of $T$ in $G_S$; so $f_S(\cdot)$ is PH-$\lceil \frac{m}{2} \rceil$.
We argue that $f(T) = \max_{S}\{f_S(T)\}$ and so $f(\cdot)$ is MPH-$\lceil \frac{m}{2} \rceil$.
It is sufficient to show that for every set $S$ the following two properties hold: (1) $f_S(S) = f(S)$; (2) $f_S(T) \leq f(T)$ for every set $T$.

The first property holds since
$f_{S}(S) = w_S(S') +w_S(S\setminus S') = f(S') +f(S)-f(S')=f(S)$.
The proof of the second property is divided into four cases:
	\begin{enumerate}
	\setlength{\itemsep}{0pt}
	\item If $S \subseteq T$ then $f_{S}(T) = w(S') +w(S\setminus S') = f(S) \leq f(T)$.
	\item If $S \not \subseteq T$ and $S' \subseteq T$ then $S\setminus S' \not \subseteq T$, therefore $f_{S}(T) = w(S') = f(S') \leq f(T)$.
	\item If $S \not \subseteq T$ and $S \setminus S' \subseteq T$ then $S'\not \subseteq T$, therefore $f_{S}(T) = w(S \setminus S') = f(S)-f(S') \leq f(S \setminus S') \leq f(T)$, where the first inequality is due to the subadditivity of $f$.
	\item If $S \setminus S', S' \not\subseteq T$ then $f_{S}(T) = 0  \leq f(T)$.	
	\end{enumerate}
\end{proof}

\bibliographystyle{alpha}
\bibliography{masterBib}

\newcommand{\etalchar}[1]{$^{#1}$}
\begin{thebibliography}{DMSW15}

\bibitem[ABDR12]{AbrahamBDR12}
Ittai Abraham, Moshe Babaioff, Shaddin Dughmi, and Tim Roughgarden.
\newblock Combinatorial auctions with restricted complements.
\newblock In {\em 13th ACM Conference on Electronic Commerce (EC)}, 2012.

\bibitem[Alo86]{Alon86}
Noga Alon.
\newblock Explicit construction of exponential sized families of
  $k$-independent sets.
\newblock {\em Discrete Mathematics}, 58:191--193, 1986.

\bibitem[ANRW15]{AlonNRW15}
Noga Alon, Noam Nisan, Ran Raz, and Omri Weinstein.
\newblock Welfare maximization with limited interaction.
\newblock In {\em {IEEE} 56th Annual Symposium on Foundations of Computer
  Science, {FOCS} 2015, Berkeley, CA, USA, 17-20 October, 2015}, pages
  1499--1512, 2015.

\bibitem[Ass17]{Assadi17}
Sepehr Assadi.
\newblock Combinatorial auctions do need modest interaction.
\newblock In {\em Proceedings of the 2017 {ACM} Conference on Economics and
  Computation, {EC} '17, Cambridge, MA, USA, June 26-30, 2017}, pages 145--162,
  2017.

\bibitem[BBK{\etalchar{+}}16]{BrodyBKLSV16}
Joshua Brody, Harry Buhrman, Michal Kouck{\'{y}}, Bruno Loff, Florian Speelman,
  and Nikolai~K. Vereshchagin.
\newblock Towards a reverse newman's theorem in interactive information
  complexity.
\newblock {\em Algorithmica}, 76(3):749--781, 2016.

\bibitem[BDF{\etalchar{+}}10]{BuchfuhrerDFKMPSSU10}
David Buchfuhrer, Shaddin Dughmi, Hu~Fu, Robert Kleinberg, Elchanan Mossel,
  Christos~H. Papadimitriou, Michael Schapira, Yaron Singer, and Christopher
  Umans.
\newblock {Inapproximability for VCG-Based Combinatorial Auctions}.
\newblock In {\em Proceedings of the Twenty-First Annual ACM-SIAM Symposium on
  Discrete Algorithms (SODA)}, 2010.

\bibitem[BG14]{BravermanG14}
Mark Braverman and Ankit Garg.
\newblock Public vs private coin in bounded-round information.
\newblock In {\em Automata, Languages, and Programming - 41st International
  Colloquium, {ICALP} 2014, Copenhagen, Denmark, July 8-11, 2014, Proceedings,
  Part {I}}, pages 502--513, 2014.

\bibitem[BGPW13]{BravermanGPW13}
Mark Braverman, Ankit Garg, Denis Pankratov, and Omri Weinstein.
\newblock From information to exact communication.
\newblock In {\em 45th}, pages 151--160, 2013.

\bibitem[BJKS04]{Bar-YossefJKS04}
Ziv Bar{-}Yossef, T.~S. Jayram, Ravi Kumar, and D.~Sivakumar.
\newblock An information statistics approach to data stream and communication
  complexity.
\newblock {\em J. Comput. Syst. Sci.}, 68(4):702--732, 2004.

\bibitem[BMW18]{BravermanMW18}
Mark Braverman, Jieming Mao, and S.~Matthew Weinberg.
\newblock On simultaneous two-player combinatorial auctions.
\newblock In {\em Proceedings of the Twenty-Ninth Annual {ACM-SIAM} Symposium
  on Discrete Algorithms, {SODA} 2018, New Orleans, LA, USA, January 7-10,
  2018}, pages 2256--2273, 2018.

\bibitem[BN02]{BlumrosenN02}
Liad Blumrosen and Noam Nisan.
\newblock Auctions with severely bounded communication.
\newblock In {\em 43rd Symposium on Foundations of Computer Science {(FOCS}
  2002), 16-19 November 2002, Vancouver, BC, Canada, Proceedings}, pages
  406--415, 2002.

\bibitem[BN07]{BlumrosenN07}
Liad Blumrosen and Noam Nisan.
\newblock Combinatorial auctions.
\newblock In Noam Nisan, Tim Roughgarden, Eva Tardos, and Vijay~V. Vazirani,
  editors, {\em Algorithmic Game Theory}, chapter~11. Cambridge University
  Press, 2007.

\bibitem[BR11a]{BhawalkarR11}
Kshipra Bhawalkar and Tim Roughgarden.
\newblock Welfare guarantees for combinatorial auctions with item bidding.
\newblock In {\em Proceedings of the Twenty-Second Annual {ACM-SIAM} Symposium
  on Discrete Algorithms, {SODA} 2011, San Francisco, California, USA, January
  23-25, 2011}, pages 700--709, 2011.

\bibitem[BR11b]{BravermanR11}
Mark Braverman and Anup Rao.
\newblock Information equals amortized communication.
\newblock In {\em {IEEE} 52nd Annual Symposium on Foundations of Computer
  Science, {FOCS} 2011, Palm Springs, CA, USA, October 22-25, 2011}, pages
  748--757, 2011.

\bibitem[Bra12]{Braverman12}
Mark Braverman.
\newblock Interactive information complexity.
\newblock In {\em Proceedings of the 44th Symposium on Theory of Computing
  Conference, {STOC} 2012, New York, NY, USA, May 19 - 22, 2012}, pages
  505--524, 2012.

\bibitem[BSS10]{BuchfuhrerSS10}
Dave Buchfuhrer, Michael Schapira, and Yaron Singer.
\newblock Computation and incentives in combinatorial public projects.
\newblock In {\em Proceedings of the 11th ACM Conference on Electronic
  Commerce}, EC '10, pages 33--42, New York, NY, USA, 2010. ACM.

\bibitem[CKST16]{ChristodoulouKST16}
George Christodoulou, Annam{\'{a}}ria Kov{\'{a}}cs, Alkmini Sgouritsa, and
  Bo~Tang.
\newblock Tight bounds for the price of anarchy of simultaneous first-price
  auctions.
\newblock {\em {ACM} Trans. Economics and Comput.}, 4(2):9:1--9:33, 2016.

\bibitem[DMSW15]{DevanurMSW15}
Nikhil~R. Devanur, Jamie Morgenstern, Vasilis Syrgkanis, and S.~Matthew
  Weinberg.
\newblock Simple auctions with simple strategies.
\newblock In {\em Proceedings of the Sixteenth {ACM} Conference on Economics
  and Computation, {EC} '15, Portland, OR, USA, June 15-19, 2015}, pages
  305--322, 2015.

\bibitem[DNS10]{DobzinskiNS10}
Shahar Dobzinski, Noam Nisan, and Michael Schapira.
\newblock Approximation algorithms for combinatorial auctions with
  complement-free bidders.
\newblock {\em Math. Oper. Res.}, 35(1):1--13, 2010.
\newblock Preliminary version in STOC 2005.

\bibitem[Dob07]{Dobzinski07}
Shahar Dobzinski.
\newblock Two randomized mechanisms for combinatorial auctions.
\newblock In {\em Proceedings of the 10th International Workshop on
  Approximation and the 11th International Workshop on Randomization, and
  Combinatorial Optimization. Algorithms and Techniques}, pages 89--103, 2007.

\bibitem[Dob11]{Dobzinski11}
Shahar Dobzinski.
\newblock An impossibility result for truthful combinatorial auctions with
  submodular valuations.
\newblock In {\em Proceedings of the 43rd ACM Symposium on Theory of Computing
  (STOC)}, 2011.

\bibitem[Dob16a]{Dobzinski16a}
Shahar Dobzinski.
\newblock Breaking the logarithmic barrier for truthful combinatorial auctions
  with submodular bidders.
\newblock In {\em Proceedings of the 48th Annual ACM SIGACT Symposium on Theory
  of Computing}, STOC 2016, pages 940--948, New York, NY, USA, 2016. ACM.

\bibitem[Dob16b]{Dobzinski16b}
Shahar Dobzinski.
\newblock Computational efficiency requires simple taxation.
\newblock In {\em 57th}, pages 209--218, 2016.

\bibitem[DSS15]{DanielySS15}
Amit Daniely, Michael Schapira, and Gal Shahaf.
\newblock Inapproximability of truthful mechanisms via generalizations of the
  {VC} dimension.
\newblock In {\em Proceedings of the Forty-Seventh Annual {ACM} on Symposium on
  Theory of Computing, {STOC} 2015, Portland, OR, USA, June 14-17, 2015}, pages
  401--408, 2015.

\bibitem[DV12a]{DobzinskiV12a}
Shahar Dobzinski and Jan Vondr{\'{a}}k.
\newblock The computational complexity of truthfulness in combinatorial
  auctions.
\newblock In {\em Proceedings of the ACM Conference on Electronic Commerce
  (EC)}, 2012.

\bibitem[DV12b]{DobzinskiV12b}
Shahar Dobzinski and Jan Vondr{\'a}k.
\newblock From query complexity to computational complexity.
\newblock In {\em Proceedings of the 44th Symposium on Theory of Computing
  (STOC)}, 2012.

\bibitem[DV13]{DobzinskiV13}
Shahar Dobzinski and Jan Vondr{\'{a}}k.
\newblock Communication complexity of combinatorial auctions with submodular
  valuations.
\newblock In {\em 24th}, pages 1205--1215, 2013.

\bibitem[Fei06]{Feige06}
Uriel Feige.
\newblock On maximizing welfare when utility functions are subadditive.
\newblock In {\em Proceedings of the 38th Annual {ACM} Symposium on Theory of
  Computing, Seattle, WA, USA, May 21-23, 2006}, pages 41--50, 2006.

\bibitem[Fei09]{Feige09}
Uriel Feige.
\newblock On maximizing welfare when utility functions are subadditive.
\newblock {\em {SIAM} J. Comput.}, 39(1):122--142, 2009.
\newblock Preliminary version in STOC 2006.

\bibitem[FFGL13]{FeldmanFGL13}
Michal Feldman, Hu~Fu, Nick Gravin, and Brendan Lucier.
\newblock Simultaneous auctions are (almost) efficient.
\newblock In {\em Symposium on Theory of Computing Conference, STOC'13, Palo
  Alto, CA, USA, June 1-4, 2013}, pages 201--210, 2013.

\bibitem[FFI{\etalchar{+}}15]{FeigeFIILS15}
Uriel Feige, Michal Feldman, Nicole Immorlica, Rani Izsak, Brendan Lucier, and
  Vasilis Syrgkanis.
\newblock A unifying hierarchy of valuations with complements and substitutes.
\newblock In {\em 29th}, pages 872--878, 2015.

\bibitem[FFT16]{FeigeFT16}
Uriel Feige, Michal Feldman, and Inbal Talgam{-}Cohen.
\newblock Oblivious rounding and the integrality gap.
\newblock In {\em 19th}, pages 8:1--8:23, 2016.

\bibitem[FV10]{FeigeV10}
Uriel Feige and Jan Vondr{\'{a}}k.
\newblock The submodular welfare problem with demand queries.
\newblock {\em Theory of Computing}, 6(1):247--290, 2010.

\bibitem[KN97]{KushilevitzN97}
Eyal Kushilevitz and Noam Nisan.
\newblock {\em Communication complexity}.
\newblock Cambridge University Press, 1997.

\bibitem[KS73]{KleitmanS73}
Daniel~J. Kleitman and Joel Spencer.
\newblock Families of $k$-independent sets.
\newblock {\em Discrete Mathematics}, 6(3):255--262, 1973.

\bibitem[LS05]{LaviS05}
Ron Lavi and Chaitanya Swamy.
\newblock Truthful and near-optimal mechanism design via linear programming.
\newblock In {\em Proceedings of the 46th Annual IEEE Symposium on Foundations
  of Computer Science (FOCS)}, 2005.

\bibitem[MSV08]{MirrokniSV08}
Vahab~S. Mirrokni, Michael Schapira, and Jan Vondr{\'{a}}k.
\newblock Tight information-theoretic lower bounds for welfare maximization in
  combinatorial auctions.
\newblock In {\em Proceedings 9th {ACM} Conference on Electronic Commerce
  (EC-2008), Chicago, IL, USA, June 8-12, 2008}, pages 70--77, 2008.

\bibitem[New91]{Newman91}
Ilan Newman.
\newblock Private vs. common random bits in communication complexity.
\newblock {\em Inf. Process. Lett.}, 39(2):67--71, 1991.

\bibitem[Nis00]{Nisan00}
Noam Nisan.
\newblock Bidding and allocation in combinatorial auctions.
\newblock In {\em {EC}}, pages 1--12, 2000.

\bibitem[NR01]{NisanR01}
Noam Nisan and Amir Ronen.
\newblock Algorithmic mechanism design.
\newblock 35:166--196, 2001.

\bibitem[NS06]{NisanS06}
Noam Nisan and Ilya Segal.
\newblock The communication requirements of efficient allocations and
  supporting prices.
\newblock 129:192--224, 2006.

\bibitem[PSS08]{PapadimitriouSS08}
Christos~H. Papadimitriou, Michael Schapira, and Yaron Singer.
\newblock On the hardness of being truthful.
\newblock In {\em Proceedings of the 49th Annual IEEE Symposium on Foundations
  of Computer Science (FOCS)}, 2008.

\bibitem[Rou14]{Roughgarden14}
Tim Roughgarden.
\newblock Barriers to near-optimal equilibria.
\newblock In {\em 55th}, pages 71--80, 2014.

\bibitem[Von08]{Vondrak08}
Jan Vondr{\'{a}}k.
\newblock Optimal approximation for the submodular welfare problem in the value
  oracle model.
\newblock In {\em Proceedings of the 40th Annual {ACM} Symposium on Theory of
  Computing, Victoria, British Columbia, Canada, May 17-20, 2008}, pages
  67--74, 2008.

\end{thebibliography}

\end{document}